%% file: Arxiv_paper.tex
\documentclass{article}
\usepackage[utf8]{inputenc} 
\usepackage[T1]{fontenc}    
\usepackage{hyperref}       
\usepackage{url}            
\usepackage{booktabs}       
\usepackage{amsfonts}       
\usepackage{nicefrac}       
\usepackage{microtype}      
\usepackage{geometry}
\usepackage{subfig}
\usepackage{bbm}  

\usepackage[round]{natbib}

\input{prel.tex}

\usepackage{todonotes}

\newcommand{\kpar}{\texttt{$k$-means||}}
\newcommand{\cost}{\mathsf{cost}}
\newcommand{\coutput}{C_{\mathrm{out}}}
\newcommand{\inum}{I}
\newcommand{\thr}{v}
\newcommand{\ctmp}{C_{\mathrm{iter}}}
\newcommand{\ka}{k_+}
\newcommand{\kdelta}{d_k}

\usepackage {multirow} 
\usepackage{boldline}

\newcommand{\comment}[1]{\ensuremath{\triangle~}\emph{#1}}
\newcommand{\OPT}{\mathrm{OPT}}

\newcommand{\smallp}
{C_{\mathrm{small}}^\alpha}

\newcommand{\largep}
{C_{\mathrm{large}}^\alpha}

\newcommand{\eo}{\psi}

\newcommand{\mainalgname}{\texttt{SOCCER}}

\makeatletter
\newcommand{\manuallabel}[2]{\def\@currentlabel{#2}\label{#1}}
\makeatother

\title{Fast Distributed $k$-Means with a Small Number of Rounds}

\usepackage{authblk}

\author[1]{Tom Hess}
\author[2]{Ron Visbord}
\author[3]{Sivan Sabato}
\affil[1,3]{Department of Computer Science, Ben Gurion University, Beer-Sheva, Israel}
\affil[2]{Independent researcher}
{
    \makeatletter
    \renewcommand\AB@affilsepx{: \protect\Affilfont}
    \makeatother

    \makeatletter
    \renewcommand\AB@affilsepx{, \protect\Affilfont}
    \makeatother
    \affil[1]{\normalsize \texttt{tomhe@bgu.ac.il}}
    \affil[2]{\normalsize \texttt{ronvisbord@gmail.com}}
    \affil[3]{\normalsize \texttt{sabatos@cs.bgu.ac.il}}
}
\date{}

\begin{document}

\maketitle

\begin{abstract} 
We propose a new algorithm for $k$-means clustering in a distributed setting, where the data is distributed across many machines, and a coordinator communicates with these machines to calculate the output clustering. 
  Our algorithm guarantees a cost approximation factor and a number of communication rounds that depend only on the computational capacity of the coordinator. Moreover, the algorithm includes a built-in stopping mechanism, which allows it to use fewer communication rounds whenever possible. We show both theoretically and empirically that in many natural cases, indeed $1-4$ rounds suffice. In comparison with the popular \kpar\ algorithm, our approach allows exploiting a larger coordinator capacity to obtain a smaller number of rounds. Our experiments show that the $k$-means cost obtained by the proposed algorithm is usually better than the cost obtained by \kpar, even when the latter is allowed a larger number of rounds. Moreover, the machine running time in our approach is considerably smaller than that of \kpar. Code for running the algorithm and experiments is available 
at \url{https://github.com/selotape/distributed_k_means}.
\end{abstract}

\section{Introduction} \label{sec:introduction}

Modern datasets can be very large, requiring algorithms that can handle massive amounts of data. This need drives the development of distributed algorithms, which use many machines that work in parallel to solve the given problem faster. In some cases, the data is already split among separate machines, again calling for a distributed solution. 
In this work, we study the classical problem of $k$-means clustering \cite{sebestyen1962decision} in the distributed setting. 
The goal of a $k$-means clustering algorithm is to select cluster centers from the input dataset that induce a $k$-means cost as close as possible to the smallest cost that can be obtained for the dataset. In a distributed framework, a main bottleneck in many practical settings
is communication. Most distributed algorithms run in communication rounds, where
in each round each machine performs an individual task, and the machines
synchronize and communicate after each round. The number of rounds is a crucial factor in the practical performance of distributed algorithms, since each such round requires synchronization and communication between the machines, which are costly and can cause time delays. Therefore, reducing the number of rounds as much as possible is a key goal for distributed algorithms.

We
focus on a common practical distributed computation model
\citep{ene2011fast,guha2019distributed}, in which one machine, called the
\emph{coordinator}, communicates with all other machines, while the data to
cluster is distributed among the machines. We consider the case where the
coordinator is capable of running heavier computations, while the machines are
more limited in their computation power and do not communicate among
themselves. This model is suitable, for instance, when the dataset to cluster
is partitioned between low-end mobile devices, and the coordinator is a stronger
machine. The data may be split among  the machines for the purpose
of performing the distributed computation, or it may be partitioned among the devices to begin with, for instance if each device has independently collected data points (e.g., by taking pictures using the device's camera). We do not make any assumption on the partition of the data, thus we support also non-i.i.d.~data.
We note that the model that we consider is different from the Federated Learning model \cite{yang2019federated,ghosh2019robust}, which emphasizes other requirements, such as privacy.

One of the most popular distributed clustering algorithms is
\kpar\ \citep{bahmani2012scalable}. This algorithm approximates 
the optimal $k$-means cost on the dataset up to a constant approximation factor, assuming that this cost
is bounded away from zero (see the example and discussion in
\citealt{bachem2017distributed}), and that a sufficient number of communication
rounds is performed. However, \kpar\ does not have an adaptive mechanism to decide how many communication rounds to run. Therefore, the number of rounds is usually set heuristically, in which case the guarantee for a constant approximation factor might not hold. 
Other distributed algorithms (e.g., \citealp{balcan2013distributed}) use only a single round of communication by definition, but do not scale well when the number of machines is large.

In this work, we propose the new distributed $k$-means clustering algorithm, \mainalgname\ (Sampling, Optimal Clustering Cost Estimation, Removal), which guarantees a constant approximation factor that depends only on the computational capacity of the coordinator, without requiring the optimal clustering cost to be bounded away from zero. Moreover, the algorithm automatically stops once a sufficient number of rounds has been completed, which can be much earlier than the worst-case number of rounds. We demonstrate that in many natural datasets, the number of rounds required by \mainalgname\ is much smaller than the worst-case upper bound. In particular, we prove that \mainalgname\ stops after a single round if the dataset is drawn from a high-dimensional Gaussian mixture. In addition, we prove that there are datasets such that \mainalgname\ stops after one round and obtains a constant approximation factor, while \kpar\ requires $k-1$ rounds for the same result.

We empirically compare \mainalgname\ to \kpar\ on synthetic and real datasets, showing that indeed in practical scenarios, \mainalgname\ stops after $1-4$ rounds. In contrast, \kpar\ does not have a stopping condition, and when stopped after the same or a similar number of rounds, it usually obtains a worse  clustering cost. Moreover, the machine run time of \mainalgname\ is almost always significantly smaller than that of \kpar\ for a comparable final cost.

Our technique is based on letting the coordinator run a (centralized) clustering algorithm on a limited number of points, and using this clustering to calculate an estimate of a truncated version of the optimal $k$-means cost on the dataset. This provides information to the machines that allows them to progressively remove points from their part of the dataset. When all the points are removed, the algorithm stops and calculates the final clustering from the centers selected in the centralized clustering runs. 
\mainalgname\ combines clustering approaches designed for
two different settings: The distributed setting \citep{ene2011fast} and
the online setting \citep{hess2021constant}.  
\citet{ene2011fast} iteratively samples points from the machines and then removes points that are close to them from consideration. We show that calculating a clustering on the point sample, along with a technique adapted from \citet{hess2021constant}, lead to a more accurate removal of points. This provides a practical and successful algorithm with approximation guarantees that depend only on the number of points that the coordinator can cluster.

\textbf{Our contribution~~} To summarize, \mainalgname\ is a new distributed $k$-means algorithm that is equipped with theoretical guarantees on its cost approximation factor and number of communication rounds, and requires an even smaller number of rounds in practice. Our experiments demonstrate its practical advantages in comparison with \kpar, in a distributed model which allows the coordinator to calculate a clustering on a limited number of points. Some of the proofs and experiment results are deferred to the appendices.

\section{Related work}\label{sec:related}

A naive approach to distributed clustering would be to implement a centralized algorithm in a straightforward manner under the distributed model. However, this tends to be impractical, since it requires a large number of communication rounds (see, e.g., the discussion in \citealp{bahmani2012scalable}). Therefore, algorithms that are specifically tailored to the distributed setting have been suggested.
Many of the algorithms   that we mention below select more than $k$ centers. It is then standard to use a weighted centralized $k$-means algorithm to reduce the number of centers to exactly $k$.  It is known \citep[e.g.,][Theorem 4]{guha2003clustering} that this preserves approximation guarantees up to constants.

One  common technique used in many distributed algorithms has the following structure: Each machine calculates a set of representatives of its own data (sometimes called \emph{coresets}). 
these are then sent to the coordinator, which uses them to calculate a set of centers  \citep{ailon2009streaming,balcan2013distributed,feldman2020turning,bachem2017practical}. These algorithms require a small constant number of communication rounds. However, the technique has the drawback that the run time and the memory size of the coordinator increase with the number of machines after suppressing a certain threshold, while in \mainalgname\, the running time of the coordinator and the machines improves linearly with the number of the machines (see experiments in \citealt{bahmani2012scalable}). Some works address the setting of distributed $k$-means with outliers \citep{guo2018distributed,guha2019distributed,chen2018practical}. These algorithms also require coordinator resources that increase with the number of machines. 
 Other distributed algorithms obtain superior guarantees, but under strong structural assumptions on the data, such as a small aspect ratio or perturbation-resilient instances \citep{voevodski2021large}, or on the partition of the data into machines \citep{bhaskara2018distributed}.

 As mentioned above, one of the most successful distributed
 $k$-means algorithms to date is \kpar
 \citep{bahmani2012scalable}, which is widely used in practice (e.g., in the
 \texttt{MLLib} library of Apache Spark, \citealp{meng2016mllib}) and also has
 theoretical guarantees. \kpar\ proposes a distributed seeding algorithm that selects a small number of potential centers. 
 The worst-case number of rounds of \kpar\ is
 $O(\log(n/\mathsf{opt}))$, where $\mathsf{opt}$ is the optimal $k$-means cost
 of the dataset. This guarantee requires the optimal $k$-means cost
 to be bounded away from zero (see also the example in
 \citealp{bachem2017distributed}).  \citet{bachem2017distributed} show that if the
variance of the dataset is bounded and $\mathsf{opt}$ is bounded away from zero, then
 \kpar\ can be stopped after a constant number of rounds. However, this requires additional information about the dataset.

 \citet{ene2011fast} proposed a distributed $k$-median algorithm (which can easily be adapted to $k$-means) with a number of communication rounds  that depends on the memory size of the coordinator.
  In each round, each machine draws two random sub-samples from its data, and sends them to the coordinator.  The coordinator adds the first  sample from each machine to the output  clustering, and uses the second sample to calculate a threshold using a simple quantile statistic. 
Then, the threshold and most of the points received by the coordinator are sent to all the machines. Each machine then removes from its dataset the points whose distance to the current clustering does not exceed the threshold. The total number of removed points is by definition a fixed fraction of the dataset. 
 The final round occurs when the remaining points in the machines fit entirely in the coordinator memory.  \citet{chen2016communication} proposed a variation on this idea that reduces the total communication, while increasing the number of rounds. \citet{kumar2015fast} generalized this technique to other related problems.  Despite its theoretical guarantees, the algorithm  of \cite{ene2011fast} has significant disadvantages in practice. First, it always uses the worst-case number of rounds. In addition, in practice, on reasonable dataset sizes, the sub-sampling does not significantly reduce the number of points relative to the original dataset. This means that calculating the final clustering is not much easier than calculating a clustering on the original dataset. In addition, the number of points sent from the coordinator to the machines is large, leading to both to a heavy communication requirement and a heavy computation in each machine. We show below how the approach of \mainalgname\ avoids these issues.

Most existing algorithms for distributed
clustering can be applied to both the $k$-means and $k$-medoids formulations with slight adaptations, and so the same
body of work is generally relevant for both formulations. The difference in formulation between $k$-means and $k$-medoids is manifested in the constant approximation factor when using a black box offline clustering algorithm.

\section{Setting and notation}
 \label{sec:setting}

For an integer $l$, denote $[l] := \{1,\ldots,l\}$.
Let $(X,\rho)$ be a finite metric space, where $X$ is a set of size $n$ and
\mbox{$\rho:X \times X \rightarrow \reals_+$} is a metric. For a point $x \in X$ and a set $T \subseteq X$, let $\rho(x,T) := \min_{y \in T} \rho(x,y)$. For simplicity, we use set notations for datasets, although they can include duplicates.
For an integer $k  \geq 2$, a $k$-clustering of $X$ is a set of (at most) $k$ points from $X$ which represent cluster centers. 
Given a set $S \subseteq X$, the $k$-means cost of $T$ on $S$ is  $\cost(S,T):= \sum_{x \in S} \rho(x, T)^2$. The goal when clustering $X$ is to find a clustering $T$ with a low cost $\cost(X,T)$. We denote by $\OPT$ an optimal $k$-means clustering: $\OPT \in \argmin_{T  \subseteq X, |T| \leq k}\cost(X,T)$.

 A (centralized) $k$-means algorithm $\cA$ takes as input a finite set of points $S$ and the parameter $k$, and outputs a $k$-clustering of $S$, denoted $\cA(S,k)$. 
For $\beta \geq 1$, $\cA$ is a \emph{$\beta$-approximation} $k$-means algorithm on $(X,\rho)$, if for all input sets $S \subseteq X$, $\cost(S,\cA(S,k)) \leq \beta \cdot \cost(S, \OPT_S)$, where $\OPT_S$ is an optimal solution on $S$ with centers from $S$: $\OPT_S \in \argmin_{T  \subseteq S, |T| \leq k}\cost(S,T)$.  In the centralized setting, the best known approximation constant for an efficient $k$-means algorithm is $9$ for a general metric space  and $6.357$ for Euclidean spaces \citep{ahmadian2019better}. 

In the coordinator model \citep{guha2019distributed} which we study, 
 the data $X$ is arbitrarily partitioned among $m$ machines, where $X_j$ denotes the set of points in machine $j\in[m]$. The machines communicate directly only with the coordinator. Broadcasts from the coordinator to the machines are counted as a single transmission. The computation is conducted in rounds, where in each round the machines perform an individual task and then communicate with the coordinator.

\section{The guarantees of  \mainalgname}\label{sec:theorem}
In this section, we present the guarantees of \mainalgname, our new distributed $k$-means algorithm, which is described in detail in \secref{main_structure}.
Similarly to \citet{ene2011fast}, we assume a bound of $\tilde{\Theta}(kn^\epsilon)$ on the number of points that can be stored in the memory of the coordinator, where $\epsilon \in (0,1)$ is a parameter linking the coordinator size with the dataset size. Specifically, we assume that the coordinator
can calculate a (centralized) clustering over a dataset of size $\eta (\epsilon) = 36kn^{\epsilon} \log (\frac{1.1k}{\delta \epsilon})$, and can store the same order of magnitude of data points. We further assume it has access to a centralized black-box $k$-means algorithm $\cA$ that can be used for this purpose.
The clustering is used by \mainalgname\ to calculate an estimate of a truncated version of the optimal attainable $k$-means cost for the dataset $X$. 
As mentioned above, most distributed algorithms select more than $k$ centers.
This number can then reduced to $k$ using a standard weighted clustering technique. \mainalgname\ selects only slightly more than $k$ centers, making the final reduction step easier. 
The following theorem gives the guarantees of  \mainalgname. The proof is provided in \secref{analysis} and the appendices referenced there.

\begin{theorem}
\label{thm:main_theorem1} 
Suppose that the size of the dataset $X$ is a sufficiently large $n$. Suppose that \mainalgname\ runs with a confidence parameter $\delta \in (0,1)$, a coordinator parameter $\epsilon \in (0,1)$, and number of centers $k\geq 5$, and suppose that the black-box algorithm $\cA$ is a $\beta$-approximation $k$-means algorithm. Denote the total number of communication rounds until \mainalgname\ stops by $I$, and denote the set of cluster centers it selects by $\coutput$. Then, with probability at least $1-\delta$,
\begin{itemize}
\item  $\inum < \frac{1}{\epsilon}-1$;
\item  $|\coutput| \leq \inum \cdot (k+ 9\log \frac{1.1k}{\delta \epsilon})$;
\item $\cost(X,\coutput) \leq   \inum \cdot (80\beta+ 44) \cdot \cost(X,\OPT)$;
\item The total number of points transmitted to the coordinator is at most $\inum \cdot \eta(\epsilon)= 72\inum  k n^{\epsilon}\log (\frac{1.1k}{\delta \epsilon})$. 
\item The total number of points broadcasted from the coordinator is at most $\inum \cdot (k+ 9\log \frac{1.1k}{\delta \epsilon})$. 
\end{itemize}
\end{theorem} 
We note that while the theorem above lists specific constants, these are in fact interdependent, so that, for instance, one can allow a larger coordinator memory constant, to obtain a significantly smaller cost approximation constant; see also the discussion in \secref{analysis}.

Before presenting the algorithm, we compare the guarantees above to the closest relevant results. In comparison with the algorithm of \citet{ene2011fast} (henceforth EIM11), \mainalgname\ uses the same number of communication rounds in the worst case. However, as will be made evident below, unlike EIM11, it can use considerably fewer rounds on many natural datasets. \mainalgname\ selects $\tilde{O}(k)$ centers, and these are all the points that the coordinator ever broadcasts to the machines. In contrast, EIM11 selects $\Omega(k n^\epsilon \log(n))$ centers and broadcasts all of them, thus its total communication to the machines is significantly larger. This also affects the computation resources required from the machines, as discussed in more detail in \secref{main_structure}. Like \mainalgname, EIM11 also obtains a constant approximation factor. While its approximation constant is smaller, the issues mentioned above make the algorithm impractical, as we observe in \secref{experiments}.

To compare these guarantees to \kpar, note that the worst-case number of rounds of \kpar\ is
$O(\log(n/\mathsf{opt}))$,
while in the theorem above (as in \citealp{ene2011fast}) it is
$1/\epsilon = \tilde{O}(\log(n)/\log(L/k))$,
where
$L$
is the limitation on the coordinator. If
$L$
is set to $\tilde{\Theta}(k)$ then the worst-case number of rounds of \mainalgname\ is similar to that of \kpar, except that it does not require $\mathsf{opt}$ to be bounded away from zero. In addition, and unlike \kpar, in this approach a larger $L$ can be used to reduce the worst-case number of iterations. Moreover, as seen below, \mainalgname\ stops on its own when the number of rounds is sufficient for the dataset. In contrast, the actual number of rounds of \kpar\ is a hyper-parameter.
In the next section, we give the full description of \mainalgname.

\begin{algorithm}[tb]
\caption{\mainalgname} \label{alg:main_algorithm}
\begin{algorithmic}[1]
\REQUIRE $\delta \in (0,1)$ (confidence), $k\in \nats$, $n\in \nats$ (data size), $\cA$ (a centralized $k$-means algorithm), 
$\epsilon \in (0,1)$ (coordinator parameter), $m\in [n]$ (number of machines).\\
At the beginning of the run, machine $j$ holds data $X_j$, where $X := \cup_{j\in [m]} X_j$.
\STATE $\coutput \leftarrow \emptyset$, $N \leftarrow n$. \label{line:settingoutput}
\WHILE{$N> \eta(\epsilon)$} \label{line:while-iteration}
\STATE $\alpha \leftarrow \eta(\epsilon)/N$.\label{line:alpha}
\STATE  For $l \in \{1,2\}$, each machine $j$ adds each point in $X_j$ to a set $P_j^l$ with independent probability $\alpha$.
\STATE  Each machine $j$ sends $P_j^1$ and $P_j^2$ to the coordinator.
\STATE In the coordinator:
\STATE \hspace{5mm} $P_1:= \underset{j \in [m]}{\bigcup}P_j^1, P_2:= \underset{j \in [m]}{\bigcup}P_j^2$.
\STATE \hspace{5mm} $\ctmp  \leftarrow \cA(P_1,\ka)$.
\STATE \hspace{5mm} $\thr := 2\cost_{\frac{3}{2}(k+1) \kdelta  }(P_2,\ctmp)/(3 k \kdelta)$. ~~\comment{See definition in \secref{main_structure}}\label{line:calc_thresh}
\STATE \hspace{5mm} $\coutput \leftarrow \coutput \cup \ctmp$.
 \STATE \hspace{5mm} Broadcast $(\thr,\ctmp)$ to each of the machines.
 \STATE Removal: Each machine $j$ updates:\\ \hspace*{1em}$X_j \leftarrow  \{x \in X_j \mid  \rho(x,\ctmp)^2 > \thr \}$. \label{line:remove_points}
\STATE Each machine sends $N_j := |X_j|$ to the coordinator, which sets $N\leftarrow \sum_{j \in [m]} N_j$.  
 \STATE $V \leftarrow \underset{j \in [m]}{\sum}X_j$. ~~~\comment{The machines sends $|X_j|$ to the coordinator.}
\ENDWHILE
\STATE All the machines send $X_j$ to the coordinator, which sets $V \leftarrow \cup_{j \in [m]} X_j$.
\STATE The coordinator calculates  $\coutput \leftarrow \coutput \cup  \cA(V,k)$\label{line:after_while}. \label{line:last_iteration}
 \STATE \textbf{return} $\coutput$.
\end{algorithmic}
\end{algorithm}

\section{The \mainalgname\ algorithm}\label{sec:main_structure}
\mainalgname\ is listed in \algref{main_algorithm}. It uses the notations
$\ka := k + 9\log(1.1k/(\delta\epsilon)), \kdelta := 6.5\log (1.1k/(\delta\epsilon) )$.
The underlying structure of \mainalgname\ is superficially similar to that of EIM11, which was described in \secref{related}. It runs a loop, where in each iteration, each machine sends the coordinator a sub-sample of its points. The coordinator then sends data points and a threshold to all machines. Then, each machine removes from its data the points that are closer to the sent points than the threshold. This is repeated in rounds, until the number of remaining points  is small enough so they can be stored in full in the coordinator.

Despite the external similarity in structure, \mainalgname\ is crucially different from EIM11 and its variants, which send most of  the points received by the coordinator to the machines, and remove a fixed fraction of the dataset in each round. The coordinator in \mainalgname\ uses the sub-samples received from the machines as input to the centralized black-box $k$-means clustering algorithm $\cA$. Then, it calculates an estimate of the truncated $k$-means cost of the centers selected by $\cA$ on the entire dataset. This estimate is then used to calculate the threshold that the machines use to remove points from their dataset.
The method for estimating the cost and calculating the threshold is based on a technique first proposed in \citet{hess2021constant}, which addresses a different setting of (centralized) online no-substitution clustering. In that work, the estimate is used for the purpose of on-the-fly center selection, when clustering a stream of points.  Our analysis shows how this type of estimate can be used to improve performance in the distributed setting, despite its original use for a completely different purpose.

In \mainalgname, in each iteration (corresponding to a communication round), each machine $j$ creates two sub-samples from its dataset, $P_j^1$ and $P_j^2$, which are then sent to the coordinator. These sub-samples are drawn independently at random from the machine's current set of points, where their sizes are set so that the total number of points sent to the coordinator by all machines is $\eta(\epsilon)$. The coordinator merges these sub-sample pairs into the respective sets $P_1$ and $P_2$. It then calculates a $\ka$-means clustering on $P_1$ using $\cA$, denoted $\ctmp$, and calculates a threshold using the \emph{truncated} cost of $\ctmp$ on $P_2$: For two sets $S,T \subseteq X$ and an integer $l$, the $l$-truncated cost of $T$ on $S$, denoted $\cost_l(S,T)$, is the total cost of the clustering after removing the $l$ points in $S$ that incur the most cost.

The coordinator adds $\ctmp$ to the output set $\coutput$, and sends $\thr$ and $\ctmp$ to each of the machines. Then, each machine removes from its dataset all the points whose distance from $\ctmp$ is at most $\sqrt{\thr}$.
Our analysis below shows that the truncated cost can be used to lower-bound the cost of points that belong to large clusters in the optimal $k$-means clustering of $X$. As a result, points that
are $\sqrt{v}$-close to some center in
$\ctmp$ are sufficiently close to an optimal center  to guarantee the final approximation factor.

Lastly, when sufficiently many points have been removed in each machine so that the entire remaining data can be handled by the coordinator, the loop terminates and the remaining points are sent to the coordinator, which calculates a $k$-clustering on them and adds the output centers to $\coutput$.

We note that the main computational burden in the machines is to calculate the distances of the data points they store from the points broadcasted by the coordinator. Therefore, the number of broadcasted points needs to be small for this burden to be reasonable. Indeed, in \mainalgname\ this number is only $\ka = k + 9\log(1.1k/(\epsilon\delta))$. In contrast, in EIM11 this number is $9kn^\epsilon \log(n/\delta)$.  Thus, for large datasets, the computational requirements from the machines in \mainalgname\ are lighter by orders of magnitude than those of EIM11.

$\ctmp$ and $v$ are calculated similarly to the centralized online clustering algorithm of \citet{hess2021constant} mentioned above. However, our constants are significantly smaller, as a result of a tighter analysis (see \appref{main_lemma}). The improvement of the constants is of significant practical importance: 
These constants are used by  \mainalgname. If they were too large, as in
\citet{hess2021constant}, then \mainalgname\ would be impractical. For instance, in
\citet{hess2021constant}, the number of outliers removed when calculating
the truncated cost is very large. Using the same number in
\mainalgname\ would have caused the fraction of removed
points in each round to be too small, leading to a large number
of rounds. Moreover,
these constants cannot be easily changed without careful
analysis, since they are inter-dependent. Finding an appropriate assignment of constants
that makes the algorithm practical while guaranteeing
the desired behaviour requires a delicate balance of many
competing quantities.

In the next section, we state the main lemma that we prove to derive the guarantees of \mainalgname.

\section{Main lemma} \label{sec:analysis}

We now give the main lemma that allows us to prove \thmref{main_theorem1}. First, we define necessary notation.
Consider the contents of the machine datasets $\{X_j\}_{j \in [m]}$ at the beginning of iteration $i$ in line \ref{line:remove_points} of \algref{main_algorithm}, and let $V_i := \cup_{j \in [m]} X_j$. Denote the points removed at iteration $i$ by $R_i:=V_i \setminus V_{i+1}$.  Let $\ctmp^i$ and $\alpha_i$  be the values of $\ctmp$ and $\alpha$, respectively, as calculated by \mainalgname\ at iteration $i$. 
To prove \thmref{main_theorem1}, we provide the following lemma, which is proved in \appref{main_lemma}.

\begin{lemma} \label{lem:main_lemma}
  Assume that \mainalgname\ runs with the parameters given in \thmref{main_theorem1}. Let $i \leq \inum$. With probability at least $1-\delta \epsilon$, 
\begin{itemize}
\item $\cost(R_i,\ctmp^i) \leq (80\beta+ 44)\cdot\cost(V_{i},\OPT)$; 
\item  $|V_{i+1}|\leq 5.5 k \kdelta/\alpha_i$. 
\end{itemize}
\end{lemma}

The first part of this lemma shows that in round $i$, the calculated cluster $\ctmp^i$ obtains a constant approximation on all the removed points in this round. This is later used to prove the overall approximation guarantee. The second part bounds the number of remaining points in each round, which is used to upper bound the number of rounds. 
 \thmref{main_theorem1} can now be proved using the lemma. The proof is provided in \appref{mainproof}

We note that the constants in \lemref{main_lemma} and, consequently, in
 \thmref{main_theorem1}, are interdependent. In particular, increasing the
 coordinator's capacity by a constant factor would decrease the cost
 approximation constant, since a larger memory constant would allow
 $P_1$ and $P_2$ to be larger, making them more representative of the full
 data, and leading to a smaller cost approximation factor. In addition, it
 would allow reducing the threshold for removal, again improving the accuracy
 at the expense of a larger coordinator capacity.

\section{Beyond worst-case: Why \mainalgname\ can stop after fewer rounds} \label{sec:beyond_the_worst_case}

As discussed above, a main desideratum of the distributed algorithm is to use a small number of communication rounds. While the worst-case number of rounds for \mainalgname\ is $\Theta(1/\epsilon)$, it stops earlier if sufficiently many data points are removed from the machine datasets, so that the current total data size can be handled by the coordinator. If this is the case, then also the approximation factor and the number of selected centers are smaller, as can be see in \thmref{main_theorem1}.

We now show that indeed, \mainalgname\ is likely to require fewer rounds on many natural datasets.  
\mainalgname\ calculates in each round the clustering $\ctmp$ based on the sub-sample sent from each machine. 
Our analysis shows that $\ctmp$ obtains a near-optimal clustering cost on points that in the optimal solution belong to clusters that are larger than $\kdelta/\alpha = \tilde{O}(n^{1-\epsilon}/k)$. Such points will typically be sufficiently close to $\ctmp$ to be removed from the machine dataset in the removal step. The number of points in small optimal clusters can be at most $k\kdelta/\alpha = O(n^{1-\epsilon})$. In many natural cases, and in particular when $n$ is sufficiently large, the optimal solution will have even fewer points, perhaps none, in such small clusters. Thus, almost all points will be removed in the first round. 
As a simple example, consider a dataset drawn from a $k$-Gaussian mixture. The following result shows that \mainalgname\ requires a single round to cluster such a dataset. The proof is provided in \appref{gaussian-p}. 
\begin{theorem}\label{thm:k_gauss}
Let $X$ be a dataset drawn from a  $k$-spherical Gaussian mixture. For sufficiently large $d$ and $n$, if  $\epsilon \geq \log\log (n/\delta) /\log n$, then with high probability, \mainalgname\ when running on $X$ will stop after one round, and output a clustering with a constant cost approximation factor.
\end{theorem}

This property of \mainalgname\ is contrasted with EIM11, which removes the same fraction of points in each round, regardless of the structure of the data, and so never stops early. 
To compare to \kpar, recall that it has no stopping mechanism and its number of rounds is set heuristically. Moreover, the following theorem, proved in \appref{hard_instance}, shows that there are cases in which \kpar\ requires $k-1$ rounds to get any finite approximation factor, while \mainalgname\ stops after a single round and finds the optimal clustering. 
\begin{theorem} \label{thm:hard_instance}
Let $k \in \nats$. For any $n_0 \in \nats$, there exists a dataset $X$ of size $n \geq n_0$, such that if \kpar\  runs on $X$ for fewer than $k-1$ rounds, then it  does not obtain a finite multiplicative approximation factor, while with probability at least $1-\delta$, \mainalgname\ stops after a single round and returns the optimal clustering. 
\end{theorem}

The experiments below demonstrate that also in practice, in many cases \mainalgname\ requires few rounds.

\section{Experiments} \label{sec:experiments} 

We report experiments on synthetic and real datasets. The code is provided at \url{https://github.com/selotape/distributed_k_means}. 
The experiments were performed on a single multi-core machine with a standard Intel processor,  which ran the code of the coordinator and of all the machines. 
We could not run EIM11 \citep{ene2011fast} on these datasets, since, as explained in \secref{main_structure}, in this algorithm the coordinator broadcasts a very large number of points to the machines. Since each machine is required to calculate the distance from each of its data points to the broadcasted set of points, this leads to a very large machine running time. For instance, for $k=100$, $n=10^7$, and $\epsilon=0.1$, the coordinator broadcasts $72,\!000$ points to the machines in each round, compared to about $200$ points sent by \mainalgname\ and \kpar. As a result, the  machine running time of EIM11 is more than a hundred-fold larger, making this algorithm far from competitive in terms of machine run time, and impractical to run in our environment.

 Both \mainalgname\ and \kpar\ output more than $k$ centers. The output $k$ clustering was calculated using the standard weighted $k$-means approach described in \secref{related}, 
 using the $k$-means algorithm of python's \texttt{scikit-learn} \citep{scikit-learn}, which was also used as our centralized black-box $k$-means algorithm for the intermediate clustering calculations of the coordinator in \mainalgname. 
To reduce variance, we fixed the sample sizes $P_1$ and $P_2$ to be exactly an $\alpha$ fraction of the current data. The parameter $l$ of \kpar, which determines the number of points to select in each round, was set to $2k$, as in \citet{bahmani2012scalable} and in the default setting of \texttt{MLLib} \citep{meng2016mllib}.  We calculated $k$-means clusterings using each of the two algorithms, for several values of $k$, on both synthetic and real datasets. The properties of the tested datasets are listed in \tabref{datasets}.

\begin{table}[h]
  \begin{center}
 \captionof{table}{Properties of datasets}
 \label{tab:datasets}
 \begin{tabular}{|l|l|l|}
 \hline
 Dataset               & \# points & Dim.   \\ \hline
 $k$-Gaussian Mixture & 10M   & 15           \\ \hline
 Higgs               & 11M   & 28        \\ \hline
 Census1990          & 2.45M & 68           \\ \hline
 KDDCup1999          & 4.8M  & 42          \\ \hline
 BigCross             & 11.6M & 57         \\ \hline
 \end{tabular}
\end{center}
\end{table}

For \mainalgname, we set $\delta = 0.1$ in all the experiments, and tested several values of $\epsilon$. For \kpar, we tested stopping after each round between $1$ and $5$. Each experiment was repeated $10$ times; we report the average of each result. Standard deviations (reported in \appref{def_exp}) were usually smaller than $2\%$ of the reported mean.

 First, we generated for each tested $k$ a synthetic dataset drawn from a $k$-Gaussian mixture in $\reals^{15}$. The mean of each Gaussian was randomly drawn from the unit cube in $\reals^{15}$, and all Gaussian were all set to be spherical with isotropic variance $\sigma=0.001$. The weight distribution of the Gaussians in the mixture was set according to the Zipf distribution, proportionally to $i^\gamma$, where $\gamma=1.5$. Each dataset consisted of ten million points drawn from this distribution. We provide the code and seed for generating these datasets at \url{https://github.com/selotape/distributed_k_means}. We then tested the algorithms on four real-world datasets with millions of points, which were used in previous papers studying similar settings: \texttt{HIGGS}, \texttt{KDDCup1999} \citep{baldi2014searching} and \texttt{Census1990}, all from the UCI repository \citep{Dua:2019}, and \texttt{Bigcross} \citep{ackermann2012streamkm++}.

\begin{table*}[t]
\begin{small}
\begin{center}
\caption{Some of the experiment results (See \appref{def_exp} for full results). 
  Top: Comparing \mainalgname\ and \kpar\ when each is running a single round.
  Bottom: \kpar\ results for 2 and 5 rounds. The factors in parenthesis for \kpar\ results provide the ratio between the \kpar\ cost or time to the corresponding values of \mainalgname.}
\vspace{0.5em}
\label{tab:experiments_summary}
\begin{tabular}{ll|llll|ll}
                      &     & \multicolumn{4}{c}{\mainalgname, one round}                    & \multicolumn{2}{c}{\kpar, one round} \\ 
Dataset & $k$   & $\epsilon$ & $|P_1|$     & Cost                     & T (seconds)         & Cost                     & T          (seconds) \\ \hlineB{2}
 Gau                     & 25  & 0.05    & 11,316  & 150                   & 0.37 & 168 $\cdot 10^3$ (x6,340)  & 0.05 (x0.14)     \\ 
                      & 100 & 0.05    & 56,440  & 150                  & 0.68 & 1,079  $\cdot 10^3$ (x1,773) & 0.05 (x0.07)     \\ \hlineB{2}
Hig                   & 25  & 0.1     & 25,335  & 144   $\cdot 10^6$   & 0.32 & 171  $\cdot 10^6$ (x1.19) & 0.05 (x0.16)    \\
                      & 100 & 0.05    & 56,440  & 122   $\cdot 10^6$   & 0.48   & 137$\cdot 10^6$   (x1.12) & 0.06 (x0.12) \\ \hlineB{2}
Cen & 25  & 0.1  & 22,018  & 188 $\cdot 10^6$     & 0.09 & 418    $\cdot 10^6$  (x2.22)    & 0.05 (x0.56)   \\
                      & 100 & 0.1     & 109,813 & 132  $\cdot 10^6$ & 0.13 & 264   $\cdot 10^6$ (x2) & 0.05 (x0.38)     \\ \hlineB{2}
 KDD & 25  & 0.2  & 110,088 & 112    $\cdot 10^{12}$ & 0.15 & 254   $\cdot 10^{12}$ (x2.08)   & 0.06 (x0.4) \\
    & 100 & 0.2  & 549,037 & 743   $\cdot 10^{10}$     & 0.26 & 5,175  $\cdot 10^{10}$ (x6.97) & 0.06 (x0.23)  \\ \hlineB{2}
 Big & 25  & 0.1  & 25,335  & 332   $\cdot 10^{10}$   & 0.38 & 519  $\cdot 10^{10}$ (x1.56)    & 0.18 (x0.47) \\
    & 100 & 0.1 & 126,354  & 152   $\cdot 10^{10}$   & 0.53 & 241    $\cdot 10^{10}$  (x1.86)  & 0.18 (x0.34)
\end{tabular}

\vspace{0.5em}
\begin{tabular}{ll|ll|ll}
     &         & \multicolumn{2}{c}{\kpar,  2 rounds}      & \multicolumn{2}{c}{\kpar,  5 rounds}       \\
Dataset & $k$   & Cost                       & T (seconds)           & Cost                       & T (seconds) \\ \hlineB{2}
\multirow{2}{*}{Gau}  & 25  & 37,350 (x246)              & 0.33 (x0.89)  & 164 (x1.1)                 & 1.98 (x5.35)  \\
     & 100 & 25,866 (x172)              & 1.09 (x1.6)   & 167 (x1.1)                 & 7.09 (x10.4)  \\ \hlineB{2}
\multirow{2}{*}{Hig} & 25  & 153 $\cdot 10^6$ (x1.06)     & 0.31 (x0.96) & 139 $\cdot 10^6$ (x1.06)     & 1.59 (x4.96)  \\
     & 100    & 125 $\cdot 10^6$ (x1.06)   & 0.85 (x1.77)  & 115  $\cdot 10^6$  (x0.94) & 5.62 (x11.7)  \\ \hlineB{2}
\multirow{2}{*}{Cen} & 25    & 218  $\cdot 10^6$  (x1.15)   & 0.15 (x1.66) & 185  $\cdot 10^6$  (x0.98)   & 0.6 (x6.66)   \\
     & 100   & 133  $\cdot 10^6$  (x1)    & 0.31 (x2.38) & 109  $\cdot 10^6$  (x0.82) & 1.66 (x12.76) \\ \hlineB{2}
\multirow{2}{*}{KDD} & 25   & 157  $\cdot 10^{12}$  (x1.4)   & 0.23 (x1.53) & 126 $\cdot 10^{12}$  (x1.12)   & 1.03 (x6.86)  \\
     & 100   & 649 $\cdot 10^{10}$  (x0.87) & 0.54 (x2.07) & 795 $\cdot 10^{10}$  (x1.07) & 3.05 (x11.73) \\ \hlineB{2}
\multirow{2}{*}{Big} & 25   & 519  $\cdot 10^{10}$ (x1.66) & 0.18 (x0.47) & 330  $\cdot 10^{10}$ (x0.99) & 2.13 (x5.6)   \\
                     & 100  & 169  $\cdot 10^{10}$ (x1.11) & 1.09 (x2.06) & 150  $\cdot 10^{10}$ (x0.99) & 6.17 (x11.64)
\end{tabular}
\end{center}
\end{small}
\end{table*}

\tabref{experiments_summary} provides some of the results of running the algorithms on the each of the datasets. Results of \mainalgname\ for all values of $\epsilon$ and for \kpar\ after all rounds between $1$ and $5$ are reported in full in \appref{kmeans_results}.
In \tabref{experiments_summary} (Top), we report the value of $\epsilon$ and the induced coordinator clustering size $|P_1|$ that resulted in \mainalgname\ stopping after a single round, and provide the obtained cost and the machine running time of \mainalgname\ and of \kpar\ after a single round. This provides a direct comparison with the same number of rounds. In \tabref{experiments_summary} (Bottom), we report the results of \kpar\ for the same experiments after two and five rounds, for comparison to \mainalgname\ after a single round. The reported machine running time was calculated by taking the sum, over all rounds, of the maximal machine running time in each round based on $50$ machines.
The   communication complexity of \mainalgname\ per round is $2|P_1|$. The communication complexity of \kpar\ per round is $l=2k$. While for large $\epsilon$, the total communication is much larger in \mainalgname, the average communication complexity per machine in \mainalgname\ is smaller, since it is $2|P_1|$ divided by the number of machines.

For the $k$-Gaussian mixtures, the first two rows of \tabref{experiments_summary} (Top) show that when the coordinator is allowed to cluster $|P_1| \approx 500\cdot k$ points, \mainalgname\ stops after a single round. In comparison, when stopping \kpar\ after one round, its resulting clustering cost is three orders of magnitude larger than that of \mainalgname. As can be seen in \tabref{experiments_summary} (Bottom), even after five rounds, the cost obtained by \kpar\ is still somewhat larger than the one obtained in one round by \mainalgname, at which point the machine running time is also larger than that of \mainalgname. For all the coordinator sizes that we tested (see \appref{kmeans_results}), the output cost of \mainalgname\ for the Gaussian mixtures was almost identical (and approximately optimal) regardless of coordinator sizes, which only affected the number of rounds.

For the other datasets, it can be seen that the cost obtained by \mainalgname\ after its single round is lower than that obtained by \kpar\ after one round, and almost always also after two rounds. In addition, the machine running time of \mainalgname\ after one round is usually significantly smaller than that of \kpar\ after running the number of rounds necessary to obtain a comparable cost.

In all of our experiments, \mainalgname\ stopped after a smaller number of rounds than the worst-case guarantee of $1/\epsilon-1$. In particular,
\tabref{dataset_small_ep} reports  experiments in which $\epsilon = 0.01$ and so the coordinator size was very small (see \appref{kmeans_results} for other values of $\epsilon$). In this case, the worst-case number of rounds is $99$, while the true number of rounds was usually between $2$ and $4$. Even with this small coordinator size, the number of rounds required by \kpar\ to obtain a comparable cost was usually much larger, as can be seen by comparing to the two rightmost columns in \tabref{dataset_small_ep}. The machine running time in \kpar\ was also usually significantly larger. Note that unlike \kpar, in which each round requires the same running time, in \mainalgname\ each additional round is considerably faster, due to the removal of points. 
Regarding the dependence of the cost upper bound on the total number of rounds of \mainalgname\ in \thmref{main_theorem1}, it can be seen in \appref{kmeans_results} that in practice the cost is similar for the same dataset for different coordinator sizes, although they each lead to a different total number of rounds.

\begin{table}[h]
\begin{small}
\begin{center}
    \captionof{table}{Results of experiments with $\epsilon=0.01$. `R' of \mainalgname\ gives the number of rounds it required. \kpar\ was run until a cost that is up to $2\%$ from that of \mainalgname.}

\label{tab:dataset_small_ep}
\scalebox{0.9}{
\begin{tabular}{ll|llll|cc}
                  &     & \multicolumn{4}{c}{\mainalgname, $\epsilon=0.01$}                    & \multicolumn{2}{c}{\kpar} \\ 
Data                                        & $k$   & $|P_1|$    & R     & Cost  & T & R & T \\  \hlineB{2}
Gau & 25  & 6,000  & 3     & 150                         & 0.72    & 15 &     12.3             \\ 
                    & 100 & 30,000 & 2     & 150  & 0.95    & 15 &    47              \\ \hlineB{2}
Hig                                      & 25  & 6,000  & 3     & 134 $\,\cdot 10^6$  & 0.6  & 8 & 3.8  \\  
                                         & 100 & 30,000 & 2    & 120 $\,\cdot 10^6$    & 0.68  & 3& 2.1  \\ \hlineB{2}
Cen                                      & 25  & 6,000  & 4     & 176  $\,\cdot 10^6$     & 0.2& 8  & 1.3 \\  
                                         & 100 & 30,000 & 3     & 110 $\,\cdot 10^6$   & 0.3  & 5 & 1.7 \\ \hlineB{2}
KDD                                      & 25  & 6,000  & 11  & 114$\,\cdot 10^{12}$       & 1 & 10&                     3.3   \\  
                                         & 100 & 30,000 & 7    & 597   $\,\cdot 10^{10}$     & 1.1  & 10 & 10.9                   \\ \hlineB{2}
Big                                & 25  & 6,000  & 3      & 319   $\,\cdot 10^{10}$  & 0.9 & 8& 5.1  \\ 
                                         & 100 & 31,000 & 2-3      & 154    $\,\cdot 10^{10}$ & 0.93  & 4 & 3.9                        \\ 
\end{tabular}}

\end{center}
\end{small}

\end{table}

We conclude that overall, if the coordinator is allowed to calculate a clustering for a moderate number of points, \mainalgname\ usually stops after a small number of rounds, and obtains a comparable or better cost than \kpar, even if the latter runs for a larger number of rounds. In addition, \mainalgname\ requires significantly less machine running time to achieve a comparable cost.

Our run time comparison focuses on machine running times, showing that in this respect \mainalgname\ is considerably faster. Our premise is that the coordinator is significantly stronger, and its computation time is not a bottleneck. However, one may still be interested in reducing the coordinator running time as well. 
To speed up the coordinator, a faster clustering implementation can be used. However, existing fast implementations are typically less successful on  more difficult datasets. We demonstrate this approach by replacing the black-box k-means implementation used for the experiments above with 
the faster \texttt{MiniBatchKMeans} implementation from \texttt{scikit-learn}. The results, reported in \appref{mini_batch_results}, show that in almost all experiments, \mainalgname\ obtains a similar cost to \kpar\ with a comparable total running time and fewer rounds. A notable exception is the \texttt{KDDCup1999} dataset. In this dataset, \texttt{MiniBatchKMeans} fails to find a clustering with a reasonable cost, even when running on the entire set of points. We believe this is because this dataset includes many outliers \cite{tavallaee2009detailed}, which are not well handled by this implementation.
This highlights the importance of using a black box that is suitable for the task at hand.

\section{Conclusion} \label{sec:conclusion}
In this work, we presented a new distributed $k$-means clustering algorithm that can require as little as one or two communication rounds, and stops on its own without having to specify the number of rounds as a parameter. Given a restriction on the maximal number of points that can be clustered by the coordinator using a centralized $k$-means algorithm, our algorithm obtains a constant approximation factor, as well as a constant upper bound  on the number of rounds. Our experiments demonstrate its effectiveness on various datasets, where it usually obtains a smaller cost than \kpar\ using fewer rounds.
We believe that the techniques used in \mainalgname\ can further be used to support robustness against outliers and machine failures, and we intend to study these challenges in future work. 

\subsubsection*{Acknowledgements}
This work was supported by the Lynn and Williams Frankel Center for Computer Science at Ben-Gurion University.

\bibliography{Mybib}
\bibliographystyle{plainnat}

\newpage
\newpage
\onecolumn
\appendix

\section{Appendix}

\subsection{Proof of \lemref{main_lemma}} \label{app:main_lemma}

In this section, we give the proof of \lemref{main_lemma}. 
We first give an auxiliary lemma. This lemma is a variation on results that were proved in \citet{hess2021constant} (henceforth abbreviated to HMS21),
 where the latter have significantly larger constants.  An additional difference is that HMS21 proved the results for $k$-median. The adaptation to $k$-means is straightforward, but affects some constants. 

\newcommand{\hkdelta}{\kdelta'}
\newcommand{\hka}{k'_+}

Denote a $k$-means solution from $X$ which is optimal for some subset $Y \subseteq X$ by the notation $\OPT_Y^X := \min_{T \subseteq X, |T| \leq k}\cost(Y,T)$. Consider the optimal clusters induced by $\OPT_Y^X$ on $Y$. HMS21 defines small optimal clusters as those optimal clusters which include at most $150\log(32k/(\delta))/\alpha$ points. 
In order to obtain guarantees with smaller constants, we use a variant of this definition. Recall the notation $\kdelta:= 6.5\log (1.1k/(\delta\epsilon) )$, $\ka := k + 9\log(1.1k/(\delta\epsilon))$. Denote $\hkdelta := 6.5\log (1.1k/\delta), \hka := k + 9\log(1.1k/\delta)$, which are equal to $\kdelta,\ka$ with $\epsilon = 1$. Let $F_\alpha(Y)$ be the $\hkdelta/\alpha$ points in $Y$ that are furthest from $\OPT_Y^X$. Define small optimal clusters to be those that of size at most $\hkdelta/\alpha$ after removing the points in $F_\alpha(Y)$. Any larger cluster is called a \emph{large optimal cluster.} 
Denote by $\smallp(Y)$ the set of points in $Y$ that belong to small optimal clusters in $Y$,  and its complement by $\largep(Y):= Y \setminus \smallp$. 
The following lemma provides results that are adaptations of results from HMS21, where latter have larger constants and hold for the original definition of small optimal clusters.  
\begin{lemma}[Adaptation of results from HMS21] \label{lem:hess_lemma}
Let $Y \subseteq X$. Let $\alpha, \delta \in (0,1)$ and set $\hka,\hkdelta$ as defined in \algref{main_algorithm}. Let $P_1, P_2 \subseteq Y$ be two independent samples of size $\alpha |Y|$, selected uniformly at random from $Y$. Let $\cA$ be a $\beta$-approximation $k$-means algorithm, and define $T:=\cA(P_1,\hka)$ .
With probability at least $1-\delta$, 
\begin{enumerate}
\item $\cost(\largep(Y)\setminus F_\alpha(Y),T) \leq (36\beta + 20) \cost(Y,\OPT)$; 
  \label{part57}

\item $\eo := \frac{2}{3 \alpha } \cost_{\frac{3}{2} (k+1) \hkdelta  }(P_2,T) \leq \cost_{(k+1)\hkdelta/\alpha}(Y,T)$; 
  \label{part58}

\item $|\{x \in Y \mid \rho(x,T)^2>\eo\alpha/(k\hkdelta) \}| \leq  5.5 k \hkdelta/\alpha$.\label{part59}  
\end{enumerate}
\end{lemma}

\begin{proof}[Proof Sketch]
  The lemma is derived by adapting results from HMS21 to our setting. The adaptation is consists of following the same proofs with minor technical differences; we give a sketch of the differences below.
  
  The three parts of the lemma are derived by adapting lemmas 5.7, 5.8, and 5.9 of HMS21 to our setting.
The original claim in Lemma 5.7 is proved for $\OPT^X_{Y} := \min_{T \subseteq X, |T| \leq k}\cost(Y,T)$, however it is easy to see that $\cost(Y,\OPT_Y^X) \leq \cost(Y,\OPT)$. In addition, the original claim does not subtract $F_\alpha(Y)$ on the LHS. This subtraction allows us to get improved final constants. 
  Lemma 5.9 gives the claim in part \ref{part59} for $Y \setminus (P_1\cup P_2)$, while in our case it holds for $Y$ (with a different constant). This is because in our case, $P_1$ and $P_2$ are independent samples, while in HMS21 they are non-overlapping.
  
  The main differences between the original lemmas and the version we give here are in the definition of large clusters and in the resulting constants.
In particular, HMS21 provided guarantees for  $\hkdelta=150\log(\frac{32k}{\delta})$ and $\hka = k+38\log(\frac{32k}{\delta})$. In the current work, we define $\hkdelta = 6.5\log(\frac{1.1k}{\delta})$ and  $\hka = k+9\log (\frac{1.1k}{\delta})$. In addition, as described above, our definition of small optimal clusters ignores the points in $F_\alpha(Y)$.
In addition to the new definition of small clusters and a tightening of the constants in the analysis, the improved constants are also due to the fact that unlike HMS21, we require fewer events to hold. For instance, we do not require the optimal points to be outside of $P_1$ and $P_2$. This allows reducing the factor in the $\log$ in the definition of $\hkdelta$.  

  The constant factor is further improved by using the tighter version of the multiplicative Chernoff bound \citep{motwani1996randomized}  to tighten the constants in Lemma 5.4 of HMS21.
Reducing the constants in $\hkdelta$ leads to a reduction in other constants as well, including those in $\hka$. An additional improvement in constants stems by assuming that $n$ is not too small, which allows avoiding certain edge cases. In particular, this allows improving the constants in the guarantees provided in HMS21 for linear bin divisions, and these affect the final result.

The approximation factor of Lemma 5.7 in HMS21 is $18\beta+10$. This factor is reduced to $9\beta +5.5$ for $k$-medians using the techniques above. For $k$-means, the triangle inequality used in several places in that proof needs to be replaced by the weak triangle inequality, leading to a final approximation factor of $36\beta + 20$.
\end{proof}

The following corollary is immediate, by applying the lemma above to the intermediate calculations in \mainalgname, and replacing $\delta,\kdelta,\ka$ by $\delta\epsilon,\hkdelta,\hka$, respectively. For simplicity, we take the sizes of $P_1$ and $P_2$ in \mainalgname\ to be exactly an $\alpha$ fraction of $V_i$. For the independent sampling mechanism of $\{P_j^l\}$  used  in \algref{main_algorithm}, this holds in expectation, and with a high probability for large data sizes, up to a negligible correction. It can also be enforced exactly and for all dataset sizes, by letting the coordinator set the number of sample points that each machine should send, based on a draw from the relevant multinomial distribution. However, since this would have a negligible effect in most cases, and makes the algorithm unnecessarily more complicated, we chose to present the simpler mechanism in \algref{main_algorithm}.

\begin{cor} \label{cor:cor-lem}
Assume that \mainalgname\ runs with the parameters given in \thmref{main_theorem1}.  Let $\alpha_i, \ctmp^i,\thr_i$ be the respective values of $\alpha, \ctmp,\eo,\thr$ calculated at iteration $i$ of \mainalgname. Let $\eo_i = \thr_ik\kdelta/\alpha_i$. Let $V_i$ be the remaining dataset at the beginning of round $i$ of \mainalgname.

With probability at least $1-\delta\epsilon$, 
\begin{itemize}
\item $\cost(\largep(V_i)\setminus F_{\alpha_i}(V_i),\ctmp^i)\leq (36\beta+20)\cost(V_i,\OPT);$
\item $\eo_i \leq \cost_{(k+1)\kdelta/\alpha_i}(V_i,\ctmp^i)$;

\item $|V_{i+1}| \leq  5.5 k\kdelta/\alpha_i$.
\end{itemize}
\end{cor} 
We now use this corollary to prove \lemref{main_lemma}.
\begin{proof} [Proof of \lemref{main_lemma}]
  To prove the first part of the lemma, we separately bound the cost of $R_i \cap (\largep(V_i)\setminus F_{\alpha_i}(V_i))$ and $R_i \cap (\smallp(V_i) \cup  F_{\alpha_i}(V_i))$ with respect to $\ctmp^i$. 
     For the first part, the bound follows from part 1 of \corref{cor-lem}, since 
\begin{align} 
&\cost(R_i \cap (\largep(V_i)\setminus F_{\alpha_i}(V_i)),\ctmp^i) \notag \\
&\leq \cost(\largep(V_i)\setminus F_{\alpha_i}(V_i),\ctmp^i) \notag \\
&\leq (36\beta + 20) \cost(V_i,\OPT). \label{eq:eq_main_lemma_1}
\end{align}
Next, we consider the second part.
Note that by the definition of small clusters, $|\smallp(V_i)\cup F_{\alpha_i}(V_i)|\leq (k+1)\kdelta/\alpha_i$. Hence, we get 
\begin{align*}
&\cost_{(k+1)\kdelta/\alpha_i}(V_i,\ctmp^i) \\
&\leq  \cost(V_i \setminus (\smallp(V_i)\cup F_{\alpha_i}(V_i)),\ctmp^i)\\
&=\cost(\largep(V_i)\setminus F_{\alpha_i}(V_i),\ctmp^i).
\end{align*}
Hence, by combing the above equation with the first and second parts of \corref{cor-lem}, we get that 
\begin{align*}
\eo_i \leq (36\beta + 20) \cost(V_i,\OPT).
\end{align*}

 Note that by the definition of $R_i$ and the equation above,
\begin{align*}
  &\forall x \in R_i,\quad \rho(x, \ctmp^i)\leq \thr^i=\frac{\eo_i}{k\kdelta/\alpha} \\
  & \leq \frac{(36\beta + 20)\cost(V_i,\OPT)}{k\kdelta/\alpha}.
\end{align*}

Hence, 
\begin{align}
  &\cost(R_i\cap (\smallp(V_i) \cup  F_{\alpha_i}(V_i)),\ctmp^i)\\
  &\leq |R_i \cap (\smallp(V_i) \cup  F_{\alpha_i}(V_i))| \frac{(36\beta + 20)\cost(V_i,\OPT)}{k\kdelta/\alpha} \notag \\ 
&\leq \frac{k+1}{k}(36\beta + 20)\cost(V_i,\OPT) \notag \\ &\leq (44\beta + 24) \cost(V_i,\OPT), \label{eq:eq_main_lemma_2}
\end{align}

Where the second inequality follows since,  \[|(\smallp(V_i) \cup  F_{\alpha_i}(V_i))| < (k+1) \kdelta/\alpha,\] and the third inequality follows since  $k \geq 5$.
By combining \eqref{eq_main_lemma_1} and \eqref{eq_main_lemma_2}, we get that
\begin{align*}
&\cost(R_i,\ctmp^i)= \\ 
&\cost(R_i\cap (\smallp(V_i) \cup  F_{\alpha_i}(V_i)),\ctmp^i) \\
&+ \cost(R_i\cap (\largep(V_i)\setminus F_{\alpha_i}(V_i)),\ctmp^i) \\ &\leq (80\beta+ 44)\cost(V_i,\OPT).
\end{align*}
This completes the proof of the first part of the lemma.  The second part of the lemma is the same as the third part of \corref{cor-lem}.
\end{proof}

\subsection{Proof of \thmref{main_theorem1}}\label{app:mainproof}

  \begin{proof}[Proof of \thmref{main_theorem1}] By a union bound, the event in \lemref{main_lemma} holds in all of the first $\min(I, 1/ \epsilon)$ rounds with probability at least $1-\delta$. 

  To prove the first part of the theorem, we show that under this joint event, \mainalgname\ stops after at most $1/\epsilon$ rounds. 
By the definition in line \ref{line:alpha}, $\alpha_i=\eta(\epsilon)/|V_i|$. Also, $\eta(\epsilon)=36kn^{\epsilon} \log (\frac{1.1k}{\delta \epsilon})$ and $\kdelta = 6.5\log (\frac{1.1k}{\delta \epsilon})$. Hence, by the second part of \lemref{main_lemma}, 
\begin{align*}
|V_{i+1}|\leq 5.5 k \kdelta/\alpha_i = 5.5 |V_i| k \kdelta /\eta(\epsilon) < |V_{i}|/n^\epsilon.
\end{align*}
 Since $V_1 = n$, it follows by induction that $|V_{i+1}| \leq  n^{1-i\epsilon}$. Recall that the stopping condition of the main loop of \mainalgname\ is $|V_{i+1}| \leq \eta (\epsilon)$. Clearly, we have $n^{1-i\epsilon} \leq \eta (\epsilon)$ once $i \geq (1- \frac{\log(36k\log (\frac{1.1k}{\delta \epsilon}))}{\log n})\frac{1}{\epsilon} -1 $.  Therefore, the total number of communication rounds is at most $(1- \frac{\log(36k\log (\frac{1.1k}{\delta \epsilon}))}{\log n})\frac{1}{\epsilon} < 1/\epsilon-1$. This proves the first part of the theorem. 

Next, we prove the cost approximation bound (the third part of the theorem). 

Since $\{R_i\}$ and $V_{\inum}$ are a partition of $X$ and $\ctmp^i \subseteq \coutput$ for all iterations $i$, we have
\begin{align*}
  \cost(X,\coutput)  \leq \sum_{i \in [\inum-1]}\cost(R_i,\ctmp^i) + \cost(V_{\inum},\ctmp^{\inum}),
\end{align*}
where $\ctmp^{\inum}$ is the result of the $k$-means clustering performed in line \ref{line:after_while} of \algref{main_algorithm}.

Using the first part of \lemref{main_lemma} and recalling that $\ctmp^{\inum}$ is a  $\beta$ approximation $k$-means solution for $V_{\inum}$, we get

 \begin{align*}
   \cost(X,\coutput) &\leq \sum_{i\in [\inum]}(80\beta+ 44) \cdot \cost(V_{i},\OPT)
                      \leq \sum_{i\in [\inum]}(80\beta+ 44)\cdot \cost(X,\OPT)\\
   &=\inum \cdot (80\beta+ 44) \cdot \cost(X,\OPT).
\end{align*} 
This completes the proof of the third part of the theorem.  
The second, fourth, and fifth parts follow directly from the definition of \mainalgname. This completes the proof.
\end{proof}

\subsection{Proof of \thmref{k_gauss}} \label{app:gaussian-p}

\begin{proof}[Proof of \thmref{k_gauss}]
Consider a $k$-Gaussian mixture in dimension $d$. Suppose that the Gaussians are all spherical with covariance matrix $\sigma^2 I$. 
For $Z \sim N(\mu,\sigma^2I)$, it is known \citep[see, e.g.][]{laurent2000adaptive}
 that for $\gamma > 0$, 
\begin{equation} \label{eq:gauss_ineq}
\P[||Z-\mu||_2^2\leq \sigma^2 d\cdot (1+2\sqrt{\frac{\log(\frac{1}{\gamma})}{d}} +\frac{2\log(\frac{1}{\gamma})}{d})] \geq 1-\gamma.
\end{equation}
In other words, for large values of $d$, almost all the points drawn from each Gaussian are about $\sigma \sqrt{d}$-far from the mean of the Gaussian. Thus, with high probability, the optimal $k$-clustering cost for a dataset of size $n$ is $\Theta(n\sigma^2 d)$.

Suppose that \mainalgname\ runs on a dataset drawn from this $k$-mixture. In the first iteration, \mainalgname\ calculates a $\ka$-clustering over a random sample of points from the dataset, using  the $\beta$ approximation algorithm $\cA$. Since $\ka$ does not depend on $d$, it is easy to see that for a large enough $d$, the average distance of the dataset points from any $\ka$ centers cannot be significantly smaller than the average distance of these points from their Gaussian centers. Therefore, the cost of the calculated $\ka$ clustering on the dataset is $\Theta(n\sigma^2 d)$.

We now show that \mainalgname\ stops after one round. 
First, we consider the value of $\cost_{\frac{3}{2} (k+1) \kdelta  }(P_2,\ctmp)$, which is used to calculate $\thr$ in line \ref{line:calc_thresh} of \algref{main_algorithm}. This is the cost of $\ctmp$ on $P_2$ after removing the  $\frac{3}{2} (k+1)\kdelta$  points that are furthest from $\ctmp$.  Note that $|P_2| = \alpha n$. Therefore, the fraction of points from $P_2$ disregarded in the calculation of the truncated cost is $k\kdelta/(\alpha n) = 6.5k\log(\frac{1.1k}{\delta})/(\alpha n) = 6.5k\log(\frac{1.1k}{\delta})/\eta(\epsilon) = O(n^{-\epsilon})$. Therefore, this fraction goes to zero for large $n$. 
It follows that  $\cost_{\frac{3}{2} (k+1) \kdelta  }(P_2,\ctmp) = |P_2|\Theta(\sigma^2 d) = \Theta(\alpha n\sigma^2d)$. Since $\alpha = \eta(\epsilon)/n = \Theta(n^{\epsilon-1})$, we get $\cost_{\frac{3}{2} (k+1) \kdelta  }(P_2,\ctmp) = \Theta(n^\epsilon \sigma^2 d)$. The threshold is thus $\thr = \Theta(\cost_{\frac{3}{2} (k+1) \kdelta  }(P_2,\ctmp)) =  \Theta(n^\epsilon\sigma^2 d)$.

We now show that \mainalgname\ stops after a single round, by showing that with high probability, all the points in $X$ are closer to $\ctmp$ than $v$.
From \eqref{gauss_ineq} with $\gamma=\log (1/\delta)/n$,
we get that with a probability at least $1-\log (1/\delta)/n$,  a point drawn from a Gaussian has a square distance of $O(\sigma^2 (d + \log(n/\delta)))$ to the center of the Gaussian. Hence, with probability at least $1-\delta$, this holds for all the points in the dataset. Clearly, this implies that the centers $\ctmp$ selected by $\cA$ must include centers with a square distance of $O(\sigma^2 (d + \log(n/\delta)))$ from the Gaussian mean, otherwise the $\beta$-approximation guarantee would not hold. 
 By the assumption of the theorem, 
 $\epsilon >  \log\log (n/\delta)/\log n$. Therefore, $n^{\epsilon}>\log(n/\delta)$. It follows that $v = \Omega(\sigma^2d\log(n/\delta))$. Thus, for a large enough $d$, $v$ is larger than the distance of all points from $\ctmp$. As a result, all the dataset points are removed in the first round of \mainalgname, and the algorithm completes after one round. By \thmref{main_theorem1}, this implies also that the cost of the output clustering is a constant approximation of the optimal cost. 
\end{proof}

\subsection{Proof of \thmref{hard_instance}} \label{app:hard_instance}

\begin{proof}[Proof of \thmref{hard_instance}]
  The proof is based on an example of  \citet{bachem2017distributed} of a hard instance for \kpar.
  \citet[Theorem 2]{bachem2017distributed} describes a dataset of size $2k-2$ such that for any value of the \kpar\ parameter $l$, \kpar\ requires at least $k-1$ rounds to obtain a constant approximation.
The dataset in the example includes $k$ distinct points $\{x_i\}_{i\in[k]}$, where $x_1$ has $k-1$ copies in the dataset and each of $x_2,\ldots,x_k$ appear a single time in the dataset.

To prove the claim in the theorem, we construct a dataset of size $n > n_0$ based on this example, by duplicating the above dataset $z = \ceil{n_0/(2k-2)}$ times. The number of rounds required by \kpar\ remains the same, as can be verified by following the proof of Theorem 2 in \citet{bachem2017distributed}.

In contrast, we now show that \mainalgname\ stops after one round on this dataset. Consider the sub-sample $P_1$, which is calculated in the first round of \mainalgname. For any $i \in [k]$,
\begin{align*}
\P[x_i \notin P_1] &\leq (1-\alpha)^z \leq \exp(-\alpha \cdot z)\\
&\leq  \exp(- \frac{\eta(\epsilon)}{2z(k-1)} \cdot z)\leq \exp(- \frac{\eta(\epsilon)}{2(k-1)})\\
& \leq \delta/k.
\end{align*}
The last inequality follows since $\eta(\epsilon) =36kn^{\epsilon}\log (\frac{k}{\delta})$. Therefore, with probability at least $1-\delta$, an instance of each of $\{x_i\}_{i \in [k]}$ is found in $P_1$.  Hence, the optimal clustering for $P_1$ includes all the distinct points from $X$, and has a cost of zero. As a result, also $\ctmp$ must have a cost of zero and so it also includes all the distinct points from $X$, leading to the removal of all the dataset points from each of the machines in line \ref{line:remove_points}. Therefore, \mainalgname\ stops after one round and returns an optimal clustering.
\end{proof}

\onecolumn
\subsection{Full experiment results} \label{app:def_exp}  
In this section, we provide the full results of all the experiments described in \secref{experiments}.  Each table reports experiments on one of the datasets in \tabref{datasets}.  The results are divided to two subsections. \appref{kmeans_results} shows the results of \mainalgname\ and \kpar\ when the standard Kmeans implementation is used as a black-box algorithm for \mainalgname, and \appref{mini_batch_results} shows the results when \texttt{MiniBatchKMeans} is used as the black box.

\subsection{Results for standard Kmeans as black-box for \mainalgname}\label{app:kmeans_results}

 The results for \mainalgname\ and \kpar, when  standard Kmeans algorithm used as black-box  for \mainalgname\ are provided below in \tabref{comparison_gaussians}, \tabref{comparison_higgs}, \tabref{comparison_census},  \tabref{comparison_kdd}, and \tabref{comparison_BigCross}.

\begin{table}[]
\caption{\texttt{k-GaussiansMixture} dataset experiments with Standard KMeans as black-box. 
‘T’ stands for time in seconds.}
\label{tab:comparison_gaussians}
\begin{tabular}{lllllllll} \hline 
$k$ & ALG                  & $\epsilon$ & $P_1$     & Output size & Rounds & Cost                & T (machine)   & T (Total)     \\ \hlineB{3}
\multirow{9}{*}{25}  & \multirow{4}{*}{\mainalgname} & 0.2 & 126,978 & 90      & 1 & 150.1$\pm$0            & 0.32$\pm$0.07 & 6.56$\pm$0.22  \\
    &                      & 0.1     & 25,335  & 96          & 1      & 150.2$\pm$0            & 0.44$\pm$0.08  & 2.51$\pm$0.12  \\
    &                      & 0.05    & 11,316  & 127$\pm$1      & 1      & 150.3$\pm$0            & 0.37$\pm$0.09  & 1.75$\pm$0.12  \\
    &                      & 0.01    & 5,939   & 348         & 3      & 150.1$\pm$0            & 0.73$\pm$0.07  & 4.11$\pm$0.21  \\ \clineB{2-9}{2}
    & \multirow{5}{*}{\kpar} & -       & -      & 51          & 1      & 1,688,270.3$\pm$951992.1 & 0.05$\pm$0     & 0.15$\pm$0.03  \\
    &                      & -       & -      & 101         & 2      & 37,530.5$\pm$46409      & 0.33$\pm$0.01  & 0.43$\pm$0.02  \\
    &                      & -       & -      & 151         & 3      & 196.9$\pm$18.6         & 0.76$\pm$0.03  & 0.87$\pm$0.03  \\
    &                      & -       & -      & 201         & 4      & 171.2$\pm$4.7          & 1.32$\pm$0.06  & 1.42$\pm$0.06  \\
    &                      & -       & -      & 251         & 5      & 164.4$\pm$2.1          & 1.98$\pm$0.07  & 2.1$\pm$0.07   \\ \hlineB{2}
\multirow{9}{*}{50}  & \multirow{4}{*}{\mainalgname} & 0.2 & 285,296 & 121     & 1 & 150.1$\pm$0            & 0.39$\pm$0.08 & 16.39$\pm$0.26 \\
    &                      & 0.1     & 56,924  & 127         & 1      & 150.2$\pm$0            & 0.48$\pm$0.08  & 5.03$\pm$0.2   \\
    &                      & 0.05    & 25,427  & 137$\pm$1      & 1      & 150.3$\pm$0            & 0.57$\pm$0.08  & 3.43$\pm$0.12  \\
    &                      & 0.01    & 13,344  & 346         & 2      & 150.2$\pm$0            & 0.79$\pm$0.09  & 5.02$\pm$0.15  \\ \clineB{2-9}{2}
    & \multirow{5}{*}{\kpar} & -       & -      & 101         & 1      & 1,283,640.5$\pm$558248.2 & 0.05$\pm$0     & 0.24$\pm$0.04  \\
    &                      & -       & -      & 201         & 2      & 13,399.3$\pm$11108.8    & 0.62$\pm$0.03  & 0.81$\pm$0.04  \\
    &                      & -       & -      & 301         & 3      & 211.5$\pm$6.8          & 1.45$\pm$0.04  & 1.66$\pm$0.04  \\
    &                      & -       & -      & 401         & 4      & 174.9$\pm$2.3          & 2.52$\pm$0.05  & 2.71$\pm$0.07  \\
    &                      & -       & -      & 501         & 5      & 166$\pm$1              & 3.83$\pm$0.08  & 4.04$\pm$0.12  \\ \hlineB{2}
\multirow{9}{*}{100} & \multirow{4}{*}{\mainalgname} & 0.2 & 633,271 & 177     & 1 & 150.1$\pm$0            & 0.53$\pm$0.05 & 73.11$\pm$0.73 \\
    &                      & 0.1     & 126,354 & 183         & 1      & 150.1$\pm$0            & 0.67$\pm$0.1   & 13.77$\pm$0.3  \\
    &                      & 0.05    & 56,440  & 212$\pm$41     & 1      & 150.3$\pm$0            & 0.68$\pm$0.11  & 8.22$\pm$0.35  \\
    &                      & 0.01    & 29,620  & 428$\pm$42     & 2      & 150.2$\pm$0            & 0.95$\pm$0.13  & 10.78$\pm$0.23 \\ \clineB{2-9}{2}
                     & \multirow{5}{*}{\kpar}    & -   & -      & 201     & 1 & 1,079,458.8$\pm$266,814.6 & 0.05$\pm$0.01 & 0.45$\pm$0.13  \\
    &                      & -       & -      & 401         & 2      & 25,866.5$\pm$16072.3    & 1.09$\pm$0.03  & 1.51$\pm$0.14  \\
    &                      & -       & -      & 601         & 3      & 226.9$\pm$61.2         & 2.67$\pm$0.06  & 3.09$\pm$0.1   \\
    &                      & -       & -      & 801         & 4      & 176.6$\pm$3.1          & 4.75$\pm$0.05  & 5.21$\pm$0.08  \\
    &                      & -       & -      & 1001        & 5      & 167.2$\pm$1.7          & 7.09$\pm$0.1   & 7.52$\pm$0.11  \\ \hlineB{2}
\multirow{8}{*}{200} & \multirow{3}{*}{\mainalgname} & 0.1 & 277,721 & 297$\pm$18 & 1 & 150.1$\pm$0            & 0.87$\pm$0.09 & 42.15$\pm$0.26 \\
    &                      & 0.05    & 124,053 & 371$\pm$81     & 1      & 150.3$\pm$0            & 0.93$\pm$0.12  & 21.97$\pm$0.55 \\
    &                      & 0.01    & 65,104  & 648$\pm$61     & 2      & 150.2$\pm$0            & 1.26$\pm$0.1   & 27.1$\pm$0.52  \\ \clineB{2-9}{2}
    & \multirow{5}{*}{\kpar} & -       & -      & 401         & 1      & 1,104,954$\pm$201,686.7   & 0.05$\pm$0.01  & 0.82$\pm$0.06  \\
    &                      & -       & -      & 801         & 2      & 26,593.9$\pm$9,916.1     & 2.08$\pm$0.04  & 3$\pm$0.14     \\
    &                      & -       & -      & 1201        & 3      & 218.8$\pm$14.9         & 4.97$\pm$0.09  & 5.89$\pm$0.1   \\
    &                      & -       & -      & 1601        & 4      & 175.7$\pm$1.5          & 8.72$\pm$0.1   & 9.74$\pm$0.17  \\
    &                      & -       & -      & 2001        & 5      & 167$\pm$1.8            & 13.16$\pm$0.12 & 14.18$\pm$0.17 \\ \hlineB{2}
\end{tabular}
\end{table}

\begin{table}[]
\caption{\texttt{Higgs} dataset experiments with Standard KMeans as black-box. 
‘T’ stands for time in seconds.}
\label{tab:comparison_higgs}
\begin{tabular}{lllllllll} \hline
$k$                  & ALG                     & $\epsilon$ & $P_1$      & Output size & Rounds & Cost ($\cdot 10^6$) & T (Machine) & T (Total)      \\ \hlineB{3}
\multirow{9}{*}{25}  & \multirow{4}{*}{\mainalgname} & 0.2       & 126,978 & 92          & 1      & 129$\pm$0.38                    & 0.3$\pm$0.02  & 112.01$\pm$4.91  \\
 &                      & 0.1  & 25,335 & 121  & 1 & 144$\pm$2.76 & 0.32$\pm$0.05  & 14.59$\pm$1.39   \\
 &                      & 0.05 & 11,316 & 204  & 2 & 144$\pm$1.53 & 0.31$\pm$0.03  & 9.23$\pm$0.33    \\
 &                      & 0.01 & 5,939  & 348  & 3 & 134$\pm$1.09 & 0.63$\pm$0.08  & 7.88$\pm$0.2     \\ \clineB{2-9}{2}
 & \multirow{5}{*}{\kpar} & -    & -      & 51   & 1 & 171$\pm$3.99 & 0.05$\pm$0     & 0.16$\pm$0.02    \\
 &                      & -    & -      & 101  & 2 & 153$\pm$1.47 & 0.31$\pm$0.02  & 0.43$\pm$0.03    \\
 &                      & -    & -      & 151  & 3 & 148$\pm$1.41 & 0.68$\pm$0.09  & 0.81$\pm$0.09    \\
 &                      & -    & -      & 201  & 4 & 143$\pm$0.98 & 1.06$\pm$0.05  & 1.19$\pm$0.06    \\
 &                      & -    & -      & 251  & 5 & 139$\pm$0.58 & 1.59$\pm$0.04  & 1.72$\pm$0.03    \\ \hlineB{2}
\multirow{9}{*}{50}  & \multirow{4}{*}{\mainalgname} & 0.2       & 285,296  & 122.8$\pm$0.4  & 1      & 117$\pm$0.19                    & 0.35$\pm$0.02 & 318.99$\pm$5.72  \\
 &                      & 0.1  & 56,924  & 177  & 1 & 134$\pm$1.75 & 0.38$\pm$0.04  & 54.53$\pm$3.34   \\
 &                      & 0.05 & 25,427  & 183  & 1 & 128$\pm$0.79 & 0.39$\pm$0.05  & 20.6$\pm$1.46    \\
 &                      & 0.01 & 13,344  & 346  & 2 & 124$\pm$0.51 & 0.56$\pm$0.04  & 17.16$\pm$0.76   \\ \clineB{2-9}{2}
 & \multirow{5}{*}{\kpar} & -    & -      & 101  & 1 & 153$\pm$1.66 & 0.05$\pm$0     & 0.25$\pm$0.04    \\
 &                      & -    & -      & 201  & 2 & 139$\pm$1.14 & 0.5$\pm$0.03   & 0.72$\pm$0.05    \\
 &                      & -    & -      & 301  & 3 & 133$\pm$0.63 & 1.13$\pm$0.05  & 1.38$\pm$0.08    \\
 &                      & -    & -      & 401  & 4 & 129$\pm$0.65 & 1.95$\pm$0.03  & 2.2$\pm$0.05     \\
 &                      & -    & -      & 501  & 5 & 127$\pm$0.45 & 2.96$\pm$0.1   & 3.24$\pm$0.11    \\ \hlineB{2}
\multirow{9}{*}{100} & \multirow{4}{*}{\mainalgname} & 0.2       & 633,272  & 178$\pm$1      & 1      & 106$\pm$0.09                    & 0.45$\pm$0.03 & 908.56$\pm$9.99  \\
 &                      & 0.1  & 126,354 & 283  & 1 & 131$\pm$1.86 & 0.44$\pm$0.02  & 191.25$\pm$12.93 \\
 &                      & 0.05 & 56,440  & 289  & 1 & 122$\pm$0.55 & 0.48$\pm$0.04  & 69.38$\pm$5.04   \\
 &                      & 0.01 & 29,620  & 508  & 2 & 120$\pm$0.59 & 0.68$\pm$0.07  & 55.61$\pm$2.24   \\ \clineB{2-9}{2}
                     & \multirow{5}{*}{\kpar}    & -         & -       & 201         & 1      & 137$\pm$0.85                    & 0.06$\pm$0.01 & 0.44$\pm$0.04    \\
 &                      & -    & -      & 401  & 2 & 125$\pm$0.92 & 0.85$\pm$0.03  & 1.29$\pm$0.06    \\
 &                      & -    & -      & 601  & 3 & 120$\pm$0.66 & 2.08$\pm$0.05  & 2.58$\pm$0.07    \\
 &                      & -    & -      & 801  & 4 & 117$\pm$0.5  & 3.75$\pm$0.05  & 4.29$\pm$0.1     \\
 &                      & -    & -      & 1001 & 5 & 115$\pm$0.61 & 5.62$\pm$0.08  & 6.15$\pm$0.08    \\ \hlineB{2}
\multirow{8}{*}{200} & \multirow{3}{*}{\mainalgname} & 0.1       & 277,721  & 470$\pm$20     & 1      & 119$\pm$2.93                    & 0.65$\pm$0.02 & 671.95$\pm$18.65 \\
 &                      & 0.05 & 124,053 & 496  & 1 & 115$\pm$0.51 & 0.67$\pm$0.05  & 251.17$\pm$16.4  \\
 &                      & 0.01 & 65,104  & 820  & 2 & 119$\pm$0.79 & 0.84$\pm$0.02  & 192.33$\pm$5.38  \\ \clineB{2-9}{2}
 & \multirow{5}{*}{\kpar} & -    & -      & 401  & 1 & 122$\pm$1.22 & 0.06$\pm$0     & 0.84$\pm$0.04    \\
 &                      & -    & -      & 801  & 2 & 112$\pm$0.27 & 1.69$\pm$0.07  & 2.6$\pm$0.08     \\
 &                      & -    & -      & 1201 & 3 & 108$\pm$0.26 & 4.1$\pm$0.1    & 5.13$\pm$0.1     \\
 &                      & -    & -      & 1601 & 4 & 106$\pm$0.31 & 7.13$\pm$0.15  & 8.33$\pm$0.17    \\
 &                      & -    & -      & 2001 & 5 & 104$\pm$0.21 & 11.08$\pm$0.33 & 12.42$\pm$0.31  \\ \hlineB{2}
\end{tabular}
\end{table}

\begin{table}[]
\caption{\texttt{Census1990} dataset experiments with Standard KMeans as black-box. 
‘T’ stands for time in seconds.}
\label{tab:comparison_census}  
\begin{tabular}{lllllllll} \hline
$k$                  & ALG                     & $\epsilon$ & $P_1$     & Output size & Rounds & Cost ($\cdot 10^6$) & T. Machine & T. Total      \\ \hlineB{3}
\multirow{9}{*}{25}  & \multirow{4}{*}{\mainalgname} & 0.2       & 95,908  & 90          & 1      & 172$\pm$1.34                    & 0.12$\pm$0.03 & 17.67$\pm$1.15   \\
 &                      & 0.1  & 22,018  & 121   & 1   & 188$\pm$4.89   & 0.1$\pm$0.02  & 5.4$\pm$0.3      \\
 &                      & 0.05 & 10,550  & 204   & 2   & 179$\pm$2.68   & 0.11$\pm$0.01 & 5.69$\pm$0.24    \\
 &                      & 0.01 & 5,856   & 489   & 4   & 176$\pm$1.07   & 0.23$\pm$0.02 & 8.69$\pm$0.17    \\ \clineB{2-9}{2}
 & \multirow{5}{*}{\kpar} & -    & -      & 51    & 1   & 418$\pm$133.76 & 0.05$\pm$0    & 0.16$\pm$0.04    \\
 &                      & -    & -      & 101   & 2   & 218$\pm$11.09  & 0.15$\pm$0    & 0.27$\pm$0.03    \\
 &                      & -    & -      & 151   & 3   & 199$\pm$3.39   & 0.28$\pm$0    & 0.4$\pm$0.03     \\
 &                      & -    & -      & 201   & 4   & 188$\pm$2.8    & 0.44$\pm$0.01 & 0.56$\pm$0.03    \\
 &                      & -    & -      & 251   & 5   & 185$\pm$3.03   & 0.61$\pm$0.01 & 0.76$\pm$0.06    \\ \hlineB{2}
\multirow{9}{*}{50}  & \multirow{4}{*}{\mainalgname} & 0.2       & 215,487 & 121         & 1      & 131$\pm$1.47                    & 0.11$\pm$0.01 & 49.62$\pm$3.63   \\
 &                      & 0.1  & 49,471  & 177   & 1   & 156$\pm$2.77   & 0.11$\pm$0.01 & 14.57$\pm$1.24   \\
 &                      & 0.05 & 23,704  & 266   & 2   & 140$\pm$2.39   & 0.14$\pm$0.03 & 15.61$\pm$0.65   \\
 &                      & 0.01 & 13,158  & 592   & 4   & 138$\pm$1.55   & 0.24$\pm$0.03 & 20.17$\pm$0.53   \\ \clineB{2-9}{2}
 & \multirow{5}{*}{\kpar} & -    & -      & 101   & 1   & 318$\pm$79.46  & 0.05$\pm$0    & 0.27$\pm$0.05    \\
 &                      & -    & -      & 201   & 2   & 169$\pm$6.17   & 0.2$\pm$0     & 0.41$\pm$0.02    \\
 &                      & -    & -      & 301   & 3   & 153$\pm$3.33   & 0.42$\pm$0.01 & 0.67$\pm$0.05    \\
 &                      & -    & -      & 401   & 4   & 144$\pm$1.97   & 0.66$\pm$0.01 & 0.94$\pm$0.07    \\
 &                      & -    & -      & 501   & 5   & 140$\pm$1.43   & 0.97$\pm$0.02 & 1.3$\pm$0.07     \\ \hlineB{2}
\multirow{9}{*}{100} & \multirow{4}{*}{\mainalgname} & 0.2       & 478,318 & 177         & 1      & 100$\pm$0.81                    & 0.15$\pm$0.02 & 172.23$\pm$11.82 \\
 &                      & 0.1  & 109,813 & 283   & 1   & 132$\pm$2.7    & 0.14$\pm$0.01 & 43.41$\pm$2      \\
 &                      & 0.05 & 52,616  & 378   & 2   & 110$\pm$0.87   & 0.17$\pm$0.03 & 45.09$\pm$2.1    \\
 &                      & 0.01 & 29,207  & 712   & 3   & 110$\pm$0.88   & 0.29$\pm$0.04 & 51.47$\pm$1.52   \\ \clineB{2-9}{2}
 & \multirow{5}{*}{\kpar} & -    & -      & 201   & 1   & 264$\pm$67.1   & 0.05$\pm$0    & 0.44$\pm$0.02    \\
 &                      & -    & -      & 401   & 2   & 133$\pm$2.3    & 0.31$\pm$0.01 & 0.77$\pm$0.08    \\
 &                      & -    & -      & 601   & 3   & 119$\pm$1.05   & 0.65$\pm$0.01 & 1.15$\pm$0.06    \\
 &                      & -    & -      & 801   & 4   & 112$\pm$0.97   & 1.1$\pm$0.01  & 1.69$\pm$0.13    \\
 &                      & -    & -      & 1001  & 5   & 109$\pm$0.67   & 1.67$\pm$0.03 & 2.21$\pm$0.08    \\ \hlineB{2}
\multirow{8}{*}{200} & \multirow{3}{*}{\mainalgname} & 0.1       & 241,364 & 489         & 1      & 111$\pm$1.9                     & 0.19$\pm$0.01 & 152.95$\pm$9.98  \\
 &                      & 0.05 & 115,648 & 563.2 & 1,2 & 89.7$\pm$1.13  & 0.22$\pm$0.03 & 143.42$\pm$10.41 \\
 &                      & 0.01 & 64,197  & 1130  & 3   & 87.3$\pm$0.54  & 0.41$\pm$0.03 & 147.67$\pm$6.72  \\ \clineB{2-9}{2}
 & \multirow{5}{*}{\kpar} & -    & -      & 401   & 1   & 224$\pm$49.18  & 0.05$\pm$0    & 0.92$\pm$0.14    \\
 &                      & -    & -      & 801   & 2   & 104$\pm$2.75   & 0.5$\pm$0.01  & 1.45$\pm$0.07    \\
 &                      & -    & -      & 1201  & 3   & 93.8$\pm$0.97  & 1.14$\pm$0.02 & 2.27$\pm$0.09    \\
 &                      & -    & -      & 1601  & 4   & 88.7$\pm$0.45  & 1.96$\pm$0.03 & 3.19$\pm$0.06    \\
 &                      & -    & -      & 2001  & 5   & 87$\pm$0.56    & 3.03$\pm$0.06 & 4.35$\pm$0.1    \\ \hlineB{2}
\end{tabular}
\end{table}

\begin{table}[]
\caption{\texttt{KDDCup1999} dataset experiments with Standard KMeans as black-box. 
‘T’ stands for time in seconds.}
\label{tab:comparison_kdd}
\begin{tabular}{lllllllll} \hline 
$k$                  & ALG                     & $\epsilon$ & $P_1$     & Output size & Rounds     & Cost ($\cdot 10^{12}$) & T. Machine & T. Total    \\ \hlineB{3}
\multirow{9}{*}{25} & \multirow{4}{*}{\mainalgname} & 0.2 & 110,088 & 115 & 1 & 112.79$\pm$10.71 & 0.15$\pm$0.02 & 9.25$\pm$1.84 \\
                     &                         & 0.1       & 23,590  & 236$\pm$40     & 2.2$\pm$0.4   & 118.21$\pm$18.54                 & 0.24$\pm$0.02 & 5.62$\pm$0.96  \\
                     &                         & 0.05      & 10,920  & 433         & 4          & 130.33$\pm$12.46                 & 0.35$\pm$0.03 & 6.04$\pm$0.32  \\
                     &                         & 0.01      & 5,896   & 1324$\pm$49    & 11.2$\pm$0.42 & 113.55$\pm$10.09                 & 1.01$\pm$0.09 & 13.91$\pm$0.73 \\ \clineB{2-9}{2}
                     & \multirow{5}{*}{\kpar}    & -         & -      & 51          & 1          & 253.76$\pm$34.98                 & 0.07$\pm$0    & 0.18$\pm$0.03  \\
                     &                         & -         & -      & 101         & 2          & 157.12$\pm$12.26                 & 0.23$\pm$0    & 0.34$\pm$0.03  \\
                     &                         & -         & -      & 151         & 3          & 148.23$\pm$22.31                 & 0.44$\pm$0.01 & 0.55$\pm$0.03  \\
                     &                         & -         & -      & 201         & 4          & 124.1$\pm$3.3                    & 0.71$\pm$0.01 & 0.82$\pm$0.04  \\
                     &                         & -         & -      & 251         & 5          & 126.4$\pm$11.82                  & 1.03$\pm$0.01 & 1.15$\pm$0.02  \\ \hlineB{2}
\multirow{9}{*}{50}  & \multirow{4}{*}{\mainalgname} & 0.2       & 247,347 & 171         & 1          & 21.77$\pm$1.33                   & 0.18$\pm$0.03 & 23.34$\pm$3.9  \\
                     &                         & 0.1       & 53,003  & 304         & 2          & 23.71$\pm$4.96                   & 0.3$\pm$0.02  & 11.84$\pm$0.83 \\
                     &                         & 0.05      & 24,535  & 515$\pm$70     & 3.5$\pm$0.5   & 23.95$\pm$3.82                   & 0.41$\pm$0.02 & 11.1$\pm$1.21  \\
                     &                         & 0.01      & 13,249  & 1352$\pm$62    & 8.8$\pm$0.4   & 22.98$\pm$2.19                   & 0.96$\pm$0.06 & 19.23$\pm$1.05 \\ \clineB{2-9}{2}
                     & \multirow{5}{*}{\kpar}    & -         & -      & 101         & 1          & 108.72$\pm$29.12                 & 0.07$\pm$0    & 0.25$\pm$0.02  \\
                     &                         & -         & -      & 201         & 2          & 37.17$\pm$10.96                  & 0.33$\pm$0.01 & 0.56$\pm$0.08  \\
                     &                         & -         & -      & 301         & 3          & 35.18$\pm$4.36                   & 0.69$\pm$0.01 & 0.91$\pm$0.05  \\
                     &                         & -         & -      & 401         & 4          & 35.19$\pm$5.25                   & 1.13$\pm$0.01 & 1.36$\pm$0.04  \\
                     &                         & -         & -      & 501         & 5          & 32.8$\pm$5.59                    & 1.69$\pm$0.02 & 1.97$\pm$0.08  \\ \hlineB{2}
\multirow{9}{*}{100} & \multirow{4}{*}{\mainalgname} & 0.2       & 549,037 & 277         & 1          & 7.43$\pm$0.66                    & 0.27$\pm$0.04 & 71.35$\pm$7.87 \\
                     &                         & 0.1       & 117,651 & 466         & 2          & 8.07$\pm$0.76                    & 0.39$\pm$0.02 & 29.1$\pm$2.11  \\
                     &                         & 0.05      & 54,461  & 667         & 3          & 7.13$\pm$0.52                    & 0.55$\pm$0.03 & 24.85$\pm$1.15 \\
                     &                         & 0.01      & 29,409  & 1528        & 7          & 5.97$\pm$0.36                    & 1.15$\pm$0.07 & 37.81$\pm$1.3  \\ \clineB{2-9}{2}
                     & \multirow{5}{*}{\kpar}    & -         & -      & 201         & 1          & 51.75$\pm$15.63                  & 0.06$\pm$0    & 0.47$\pm$0.05  \\
                     &                         & -         & -      & 401         & 2          & 6.49$\pm$0.74                    & 0.54$\pm$0.01 & 0.92$\pm$0.04  \\
                     &                         & -         & -      & 601         & 3          & 8.41$\pm$0.8                     & 1.16$\pm$0.02 & 1.62$\pm$0.14  \\
                     &                         & -         & -      & 801         & 4          & 8.54$\pm$1.34                    & 1.99$\pm$0.03 & 2.4$\pm$0.04   \\
                     &                         & -         & -      & 1001        & 5          & 7.95$\pm$0.49                    & 3.05$\pm$0.03 & 3.55$\pm$0.13  \\ \hlineB{2}
\multirow{8}{*}{200} & \multirow{3}{*}{\mainalgname} & 0.1       & 258,592 & 778         & 2          & 3.06$\pm$0.08                    & 0.56$\pm$0.04 & 93.74$\pm$6.77 \\
                     &                         & 0.05      & 119,705 & 1088        & 3          & 2.89$\pm$0.33                    & 0.73$\pm$0.05 & 64.74$\pm$4.33 \\
                     &                         & 0.01      & 64,641  & 2060        & 6          & 2.46$\pm$0.21                    & 1.46$\pm$0.07 & 82.2$\pm$2.99  \\ \clineB{2-9}{2}
                     & \multirow{5}{*}{\kpar}    & -         & -      & 401         & 1          & 10.85$\pm$2.14                   & 0.06$\pm$0    & 0.89$\pm$0.18  \\
                     &                         & -         & -      & 801         & 2          & 1.71$\pm$0.5                     & 0.92$\pm$0.02 & 1.81$\pm$0.23  \\
                     &                         & -         & -      & 1201        & 3          & 2.62$\pm$0.35                    & 2.11$\pm$0.05 & 3.06$\pm$0.25  \\
                     &                         & -         & -      & 1601        & 4          & 2.76$\pm$0.11                    & 3.64$\pm$0.09 & 4.6$\pm$0.12   \\
                     &                         & -         & -      & 2001        & 5          & 2.41$\pm$0.35                    & 5.69$\pm$0.08 & 6.73$\pm$0.11 \\ \hlineB{2}
\end{tabular}
\end{table}

\begin{table}[]
\caption{\texttt{BigCross} dataset experiments with Standard KMeans as black-box. 
‘T’ stands for time in seconds.}
\label{tab:comparison_BigCross}
\begin{tabular}{lllllllll} \hline 
$k$                  & ALG                     & $\epsilon$ & $P_1$     & Output size & Rounds & Cost ($\cdot 10^{10}$) & T. Machine  & T. Total      \\ \hlineB{3}
\multirow{9}{*}{25}  & \multirow{4}{*}{\mainalgname} & 0.2       & 126,978 & 90          & 1      & 328$\pm$ 5       & 0.43$\pm$0.06  & 60.32$\pm$3.7    \\
                     &                         & 0.1       & 25,335  & 106$\pm$13     & 1      & 332$\pm$ 7       & 0.39$\pm$0.03  & 11.95$\pm$0.9    \\
                     &                         & 0.05      & 11,316  & 204         & 2      & 345$\pm$ 5       & 0.4$\pm$0.05   & 6.82$\pm$0.24    \\
                     &                         & 0.01      & 5,939   & 358$\pm$13     & 3      & 319$\pm$ 2       & 0.87$\pm$0.08  & 7.78$\pm$0.26    \\ \clineB{2-9}{2}
                     & \multirow{5}{*}{\kpar}    & -         & -      & 51          & 1      & 519$\pm$ 40      & 0.18$\pm$0.01  & 0.27$\pm$0.03    \\
                     &                         & -         & -      & 101         & 2      & 367$\pm$ 9       & 0.49$\pm$0.01  & 0.6$\pm$0.02     \\
                     &                         & -         & -      & 151         & 3      & 350$\pm$ 5       & 0.93$\pm$0.04  & 1.05$\pm$0.06    \\
                     &                         & -         & -      & 201         & 4      & 339$\pm$ 6       & 1.59$\pm$0.05  & 1.71$\pm$0.06    \\
                     &                         & -         & -      & 251         & 5      & 330$\pm$ 6       & 2.14$\pm$0.09  & 2.28$\pm$0.1     \\ \hlineB{2}
\multirow{9}{*}{50}  & \multirow{4}{*}{\mainalgname} & 0.2       & 285,296 & 121         & 1      & 224$\pm$ 4       & 0.41$\pm$0.01  & 164.48$\pm$16.65 \\
                     &                         & 0.1       & 56,924  & 127         & 1      & 221$\pm$ 3       & 0.47$\pm$0.05  & 35.69$\pm$2.03   \\
                     &                         & 0.05      & 25,427  & 266         & 2      & 242$\pm$ 3       & 0.5$\pm$0.05   & 19.17$\pm$1.26   \\
                     &                         & 0.01      & 13,344  & 444         & 3      & 215$\pm$ 1       & 0.83$\pm$0.06  & 17.42$\pm$0.61   \\ \clineB{2-9}{2}
                     & \multirow{5}{*}{\kpar}    & -         & -      & 101         & 1      & 365$\pm$ 28      & 0.18$\pm$0.02  & 0.39$\pm$0.03    \\
                     &                         & -         & -      & 201         & 2      & 244$\pm$ 6       & 0.78$\pm$0.07  & 1.02$\pm$0.09    \\
                     &                         & -         & -      & 301         & 3      & 230$\pm$ 2       & 1.33$\pm$0.01  & 1.58$\pm$0.03    \\
                     &                         & -         & -      & 401         & 4      & 223$\pm$ 2       & 2.27$\pm$0.07  & 2.55$\pm$0.1     \\
                     &                         & -         & -      & 501         & 5      & 217$\pm$ 1       & 3.54$\pm$0.19  & 3.84$\pm$0.12    \\ \hlineB{2}
\multirow{9}{*}{100} & \multirow{4}{*}{\mainalgname} & 0.2       & 633,272 & 177         & 1      & 151$\pm$ 1       & 0.54$\pm$0.06  & 510.31$\pm$41.78 \\
                     &                         & 0.1       & 126,354 & 183         & 1      & 152$\pm$ 1       & 0.53$\pm$0.03  & 105.06$\pm$9.78  \\
                     &                         & 0.05      & 56,440  & 289         & 1      & 170$\pm$ 2       & 0.61$\pm$0.05  & 48.38$\pm$3.31   \\
                     &                         & 0.01      & 29,620  & 580$\pm$ 50    & 2,3    & 154$\pm$ 2       & 0.94$\pm$0.07  & 48.47$\pm$3.34   \\ \clineB{2-9}{2}
                     & \multirow{5}{*}{\kpar}    & -         & -      & 201         & 1      & 242$\pm$ 20      & 0.18$\pm$0.03  & 0.56$\pm$0.06    \\
                     &                         & -         & -      & 401         & 2      & 169$\pm$ 2       & 1.1$\pm$0.04   & 1.54$\pm$0.05    \\
                     &                         & -         & -      & 601         & 3      & 157$\pm$ 2       & 2.43$\pm$0.18  & 2.94$\pm$0.18    \\
                     &                         & -         & -      & 801         & 4      & 153$\pm$ 1       & 3.94$\pm$0.25  & 4.49$\pm$0.33    \\
                     &                         & -         & -      & 1001        & 5      & 150$\pm$ 1       & 6.17$\pm$0.31  & 6.71$\pm$0.32    \\ \hlineB{2}
\multirow{8}{*}{200} & \multirow{3}{*}{\mainalgname} & 0.1       & 277,721 & 289         & 1      & 103$\pm$ 0       & 0.73$\pm$0.04  & 336.87$\pm$27.02 \\
                     &                         & 0.05      & 124,053 & 496         & 1      & 117$\pm$ 2       & 0.74$\pm$0.04  & 142.19$\pm$9.06  \\
                     &                         & 0.01      & 65,104  & 820         & 2      & 109$\pm$ 1       & 1.13$\pm$0.07  & 121.35$\pm$4.43  \\ \clineB{2-9}{2}
                     & \multirow{5}{*}{\kpar}    & -         & -      & 401         & 1      & 166$\pm$ 9       & 0.21$\pm$0.06  & 1.04$\pm$0.07    \\
                     &                         & -         & -      & 801         & 2      & 119$\pm$ 1       & 1.8$\pm$0.03   & 2.72$\pm$0.03    \\
                     &                         & -         & -      & 1201        & 3      & 111$\pm$ 1       & 4.08$\pm$0.13  & 5.31$\pm$0.28    \\
                     &                         & -         & -      & 1601        & 4      & 107$\pm$ 1       & 7.31$\pm$0.09  & 8.56$\pm$0.09    \\
                     &                         & -         & -      & 2001        & 5      & 106$\pm$ 0       & 10.86$\pm$0.32 & 12.24$\pm$0.32  \\ \hlineB{2}
\end{tabular}
\end{table} 
\clearpage

\subsection{Results for \texttt{MiniBatchKMeans} as black-box for \mainalgname}\label{app:mini_batch_results}
The results for \mainalgname\ and \kpar, when  \texttt{MiniBatchKMeans} used as black-box  for \mainalgname\ are provided below in  \tabref{experiments_summary_gaussian_S}, \tabref{experiments_summary_higgs_S} , \tabref{experiments_summary_census_S}, ,  \tabref{experiments_summary_KDD_S}, and  \tabref{experiments_summary_bigcross_S}. 

\begin{table}[]
\caption{\texttt{$k$-GaussianMixture} dataset experiments with \texttt{MiniBatchKMeans} used as black-box.
‘T’ stands for time in seconds.}
\label{tab:experiments_summary_gaussian_S}
\begin{tabular}{lllllllll} \hline 
$k$                  & ALG                     & $\epsilon$ & $P_1$     & Output size & Rounds   & Cost                & T (machine)   & T (Total)     \\ \hlineB{3}
\multirow{9}{*}{25}  & \multirow{4}{*}{\mainalgname} & 0.2     & 126,978 & 105$\pm$13     & 1        & 150.2$\pm$0.2          & 0.32$\pm$0.06  & 1.03$\pm$0.2   \\
                     &                         & 0.1     & 25,335  & 161$\pm$50     & 1.6$\pm$0.5 & 150.3$\pm$0.1          & 0.49$\pm$0.12  & 1.14$\pm$0.25  \\
                     &                         & 0.05    & 11,316  & 178$\pm$53     & 1.5$\pm$0.5 & 150.5$\pm$0.1          & 0.49$\pm$0.14  & 1.05$\pm$0.28  \\
                     &                         & 0.01    & 5,939   & 348         & 3        & 150.1$\pm$0            & 0.74$\pm$0.12  & 1.67$\pm$0.18  \\ \clineB{2-9}{2}
                     & \multirow{5}{*}{\kpar}    & -       & -      & 51          & 1        & 1,688,270.3$\pm$951992.1 & 0.05$\pm$0     & 0.15$\pm$0.03  \\
                     &                         & -       & -      & 101         & 2        & 37,530.5$\pm$46409      & 0.33$\pm$0.01  & 0.43$\pm$0.02  \\
                     &                         & -       & -      & 151         & 3        & 196.9$\pm$18.6         & 0.76$\pm$0.03  & 0.87$\pm$0.03  \\
                     &                         & -       & -      & 201         & 4        & 171.2$\pm$4.7          & 1.32$\pm$0.06  & 1.42$\pm$0.06  \\
                     &                         & -       & -      & 251         & 5        & 164.4$\pm$2.1          & 1.98$\pm$0.07  & 2.1$\pm$0.07   \\ \hlineB{2}
\multirow{9}{*}{50}  & \multirow{4}{*}{\mainalgname} & 0.2     & 285,296 & 171         & 1        & 150.4$\pm$0.3          & 0.52$\pm$0.1   & 1.91$\pm$0.28  \\
                     &                         & 0.1     & 56,924  & 216$\pm$41     & 1.5$\pm$0.5 & 152.1$\pm$5.2          & 0.51$\pm$0.07  & 1.5$\pm$0.31   \\
                     &                         & 0.05    & 25,427  & 244$\pm$42     & 1.7$\pm$0.5 & 150.6$\pm$0.2          & 0.65$\pm$0.14  & 1.48$\pm$0.23  \\
                     &                         & 0.01    & 13,344  & 361$\pm$47     & 2.1$\pm$0.3 & 150.3$\pm$0.1          & 0.83$\pm$0.09  & 1.86$\pm$0.18  \\ \clineB{2-9}{2}
                     & \multirow{5}{*}{\kpar}    & -       & -      & 101         & 1        & 1,283,640.5$\pm$558248.2 & 0.05$\pm$0     & 0.24$\pm$0.04  \\
                     &                         & -       & -      & 201         & 2        & 13,399.3$\pm$11108.8    & 0.62$\pm$0.03  & 0.81$\pm$0.04  \\
                     &                         & -       & -      & 301         & 3        & 211.5$\pm$6.8          & 1.45$\pm$0.04  & 1.66$\pm$0.04  \\
                     &                         & -       & -      & 401         & 4        & 174.9$\pm$2.3          & 2.52$\pm$0.05  & 2.71$\pm$0.07  \\
                     &                         & -       & -      & 501         & 5        & 166$\pm$1              & 3.83$\pm$0.08  & 4.04$\pm$0.12  \\ \hlineB{2}
\multirow{9}{*}{100} & \multirow{4}{*}{\mainalgname} & 0.2     & 633,271 & 277         & 1        & 628.2$\pm$612.3        & 0.67$\pm$0.13  & 4.19$\pm$0.76  \\
                     &                         & 0.1     & 126,354 & 406$\pm$52     & 2        & 154$\pm$6.1            & 0.91$\pm$0.11  & 3.22$\pm$0.3   \\
                     &                         & 0.05    & 56,440  & 390$\pm$54     & 1.9$\pm$0.3 & 150.7$\pm$0.5          & 0.94$\pm$0.1   & 2.7$\pm$0.37   \\
                     &                         & 0.01    & 29,620  & 550$\pm$54     & 2.4$\pm$0.5 & 150.3$\pm$0            & 1.05$\pm$0.13  & 2.91$\pm$0.21  \\ \clineB{2-9}{2}
                     & \multirow{5}{*}{\kpar}    & -       & -      & 201         & 1        & 1,079,458.8$\pm$266814.6 & 0.05$\pm$0.01  & 0.45$\pm$0.13  \\
                     &                         & -       & -      & 401         & 2        & 25,866.5$\pm$16072.3    & 1.09$\pm$0.03  & 1.51$\pm$0.14  \\
                     &                         & -       & -      & 601         & 3        & 226.9$\pm$61.2         & 2.67$\pm$0.06  & 3.09$\pm$0.1   \\
                     &                         & -       & -      & 801         & 4        & 176.6$\pm$3.1          & 4.75$\pm$0.05  & 5.21$\pm$0.08  \\
                     &                         & -       & -      & 1001        & 5        & 167.2$\pm$1.7          & 7.09$\pm$0.1   & 7.52$\pm$0.11  \\ \hlineB{2}
\multirow{8}{*}{200} & \multirow{3}{*}{\mainalgname} & 0.1 & 277,721 & 733$\pm$96 & 1.9$\pm$0.3 & 156.9$\pm$4.8 & 1.42$\pm$0.17 & 6.94$\pm$0.31 \\
                     &                         & 0.05    & 124,053 & 757$\pm$74     & 2        & 150.9$\pm$0.6          & 1.47$\pm$0.12  & 5.13$\pm$0.29  \\
                     &                         & 0.01    & 65,104  & 1008$\pm$86    & 3        & 855.6$\pm$1139.9       & 1.71$\pm$0.11  & 5.53$\pm$0.26  \\ \clineB{2-9}{2}
                     & \multirow{5}{*}{\kpar}    & -       & -      & 401         & 1        & 1,104,954$\pm$201686.7   & 0.05$\pm$0.01  & 0.82$\pm$0.06  \\
                     &                         & -       & -      & 801         & 2        & 26,593.9$\pm$9916.1     & 2.08$\pm$0.04  & 3$\pm$0.14     \\
                     &                         & -       & -      & 1201        & 3        & 218.8$\pm$14.9         & 4.97$\pm$0.09  & 5.89$\pm$0.1   \\
                     &                         & -       & -      & 1601        & 4        & 175.7$\pm$1.5          & 8.72$\pm$0.1   & 9.74$\pm$0.17  \\
                     &                         & -       & -      & 2001        & 5        & 167$\pm$1.8            & 13.16$\pm$0.12 & 14.18$\pm$0.17 \\ \hlineB{2}
\end{tabular}
\end{table}

\begin{table}[]
\caption{\texttt{Higgs} dataset experiments with MiniBatchKMeans used as black-box.
‘T’ stands for time in seconds.}
\label{tab:experiments_summary_higgs_S}
\begin{tabular}{lllllllll} \hline 
$k$                  & Alg                     & $\epsilon$ & $P_1$     & Output size & Rounds & Cost ($\cdot 10^6$) & T (machine)   & T (Total)     \\ \hlineB{3}
\multirow{9}{*}{25}  & \multirow{4}{*}{\mainalgname} & 0.2     & 126,978 & 92$\pm$2       & 1      & 129.5$\pm$0.9                   & 0.27$\pm$0.03  & 1.04$\pm$0.12  \\
                     &                         & 0.1     & 25,335  & 121         & 1      & 141.5$\pm$2                     & 0.29$\pm$0.02  & 0.83$\pm$0.09  \\
                     &                         & 0.05    & 11,316  & 204         & 2      & 144.5$\pm$1.8                   & 0.31$\pm$0.04  & 0.99$\pm$0.11  \\
                     &                         & 0.01    & 5,939   & 348         & 3      & 135.2$\pm$0.7                   & 0.49$\pm$0.04  & 1.52$\pm$0.19  \\ \clineB{2-9}{2}
                     & \multirow{5}{*}{\kpar}    & -       & -      & 51          & 1      & 171$\pm$4                       & 0.05$\pm$0     & 0.16$\pm$0.02  \\
                     &                         & -       & -      & 101         & 2      & 153$\pm$1.5                     & 0.31$\pm$0.02  & 0.43$\pm$0.03  \\
                     &                         & -       & -      & 151         & 3      & 148$\pm$1.4                     & 0.68$\pm$0.09  & 0.81$\pm$0.09  \\
                     &                         & -       & -      & 201         & 4      & 143$\pm$1                       & 1.06$\pm$0.05  & 1.19$\pm$0.06  \\
                     &                         & -       & -      & 251         & 5      & 139$\pm$0.6                     & 1.59$\pm$0.04  & 1.72$\pm$0.03  \\  \hlineB{2}
\multirow{9}{*}{50}  & \multirow{4}{*}{\mainalgname} & 0.2     & 285,296 & 123$\pm$0.4    & 1      & 118.1$\pm$0.3                   & 0.35$\pm$0.04  & 1.98$\pm$0.22  \\
                     &                         & 0.1     & 56,924  & 177         & 1      & 133.4$\pm$1.6                   & 0.37$\pm$0.04  & 1.26$\pm$0.18  \\
                     &                         & 0.05    & 25,427  & 183         & 1      & 128.2$\pm$0.8                   & 0.38$\pm$0.05  & 1.16$\pm$0.11  \\
                     &                         & 0.01    & 13,344  & 346         & 2      & 124.9$\pm$0.4                   & 0.54$\pm$0.06  & 1.69$\pm$0.16  \\ \clineB{2-9}{2}
                     & \multirow{5}{*}{\kpar}    & -       & -      & 101         & 1      & 153$\pm$1.7                     & 0.05$\pm$0     & 0.25$\pm$0.04  \\
                     &                         & -       & -      & 201         & 2      & 139$\pm$1.1                     & 0.5$\pm$0.03   & 0.72$\pm$0.05  \\
                     &                         & -       & -      & 301         & 3      & 133$\pm$0.6                     & 1.13$\pm$0.05  & 1.38$\pm$0.08  \\
                     &                         & -       & -      & 401         & 4      & 129$\pm$0.7                     & 1.95$\pm$0.03  & 2.2$\pm$0.05   \\
                     &                         & -       & -      & 501         & 5      & 127$\pm$0.4                     & 2.96$\pm$0.1   & 3.24$\pm$0.11  \\ \hlineB{2}
\multirow{9}{*}{100} & \multirow{4}{*}{\mainalgname} & 0.2     & 633,272 & 179$\pm$0.4    & 1      & 106.7$\pm$0.2                   & 0.55$\pm$0.03  & 4.03$\pm$0.2   \\
                     &                         & 0.1     & 126,354 & 283         & 1      & 127.5$\pm$1.4                   & 0.49$\pm$0.06  & 2.06$\pm$0.2   \\
                     &                         & 0.05    & 56,440  & 289         & 1      & 121.5$\pm$0.7                   & 0.51$\pm$0.06  & 1.89$\pm$0.18  \\
                     &                         & 0.01    & 29,620  & 508         & 2      & 121.1$\pm$0.7                   & 0.73$\pm$0.06  & 2.62$\pm$0.14  \\ \clineB{2-9}{2}
                     & \multirow{5}{*}{\kpar}    & -       & -      & 201         & 1      & 137$\pm$0.8                     & 0.06$\pm$0.01  & 0.44$\pm$0.04  \\
                     &                         & -       & -      & 401         & 2      & 125$\pm$0.9                     & 0.85$\pm$0.03  & 1.29$\pm$0.06  \\
                     &                         & -       & -      & 601         & 3      & 120$\pm$0.7                     & 2.08$\pm$0.05  & 2.58$\pm$0.07  \\
                     &                         & -       & -      & 801         & 4      & 117$\pm$0.5                     & 3.75$\pm$0.05  & 4.29$\pm$0.1   \\
                     &                         & -       & -      & 1001        & 5      & 115$\pm$0.6                     & 5.62$\pm$0.08  & 6.15$\pm$0.08  \\ \hlineB{2}
\multirow{8}{*}{200} & \multirow{3}{*}{\mainalgname} & 0.1     & 277,721 & 480$\pm$17     & 1      & 117.3$\pm$1.3                   & 0.79$\pm$0.06  & 4.23$\pm$0.31  \\
                     &                         & 0.05    & 124,053 & 496         & 1      & 115.3$\pm$0.4                   & 0.77$\pm$0.07  & 3.54$\pm$0.21  \\
                     &                         & 0.01    & 65,104  & 820         & 2      & 116.7$\pm$0.7                   & 0.94$\pm$0.08  & 4.45$\pm$0.34  \\ \clineB{2-9}{2}
                     & \multirow{5}{*}{\kpar}    & -       & -      & 401         & 1      & 122$\pm$1.2                     & 0.06$\pm$0     & 0.84$\pm$0.04  \\
                     &                         & -       & -      & 801         & 2      & 112$\pm$0.3                     & 1.69$\pm$0.07  & 2.6$\pm$0.08   \\
                     &                         & -       & -      & 1201        & 3      & 108$\pm$0.3                     & 4.1$\pm$0.1    & 5.13$\pm$0.1   \\
                     &                         & -       & -      & 1601        & 4      & 106$\pm$0.3                     & 7.13$\pm$0.15  & 8.33$\pm$0.17  \\
                     &                         & -       & -      & 2001        & 5      & 104$\pm$0.2                     & 11.08$\pm$0.33 & 12.42$\pm$0.31 \\ \hlineB{2}
\end{tabular}
\end{table}

\begin{table}[]
\caption{\texttt{Census1990} dataset experiments with MiniBatchKMeans used as black-box.
‘T’ stands for time in seconds.}
\label{tab:experiments_summary_census_S}
\begin{tabular}{lllllllll} \hline 
$k$                  & ALG                     & $\epsilon$ & $P_1$   & Output size & Rounds   & Cost ($\cdot 10^6$) & T (machine)  & T (Total)    \\ \hlineB{3}
\multirow{9}{*}{25}  & \multirow{4}{*}{\mainalgname} & 0.2     & 95,908  & 90          & 1        & 171.3$\pm$1.7                   & 0.11$\pm$0.05 & 0.96$\pm$0.14 \\
 &                      & 0.1  & 22,018  & 121  & 1 & 187.8$\pm$4   & 0.11$\pm$0.04 & 0.65$\pm$0.16 \\
 &                      & 0.05 & 10,550  & 204  & 2 & 179.6$\pm$3.6 & 0.14$\pm$0.05 & 0.85$\pm$0.13 \\
 &                      & 0.01 & 5,856   & 489  & 4 & 175.8$\pm$1.6 & 0.3$\pm$0.07  & 1.65$\pm$0.17 \\ \clineB{2-9}{2}
 & \multirow{5}{*}{\kpar} & -    & -      & 51   & 1 & 418$\pm$133.8 & 0.05$\pm$0    & 0.16$\pm$0.04 \\
 &                      & -    & -      & 101  & 2 & 218$\pm$11.1  & 0.15$\pm$0    & 0.27$\pm$0.03 \\
 &                      & -    & -      & 151  & 3 & 199$\pm$3.4   & 0.28$\pm$0    & 0.4$\pm$0.03  \\
 &                      & -    & -      & 201  & 4 & 188$\pm$2.8   & 0.44$\pm$0.01 & 0.56$\pm$0.03 \\
 &                      & -    & -      & 251  & 5 & 185$\pm$3     & 0.61$\pm$0.01 & 0.76$\pm$0.06 \\ \hlineB{2}
\multirow{9}{*}{50}  & \multirow{4}{*}{\mainalgname} & 0.2     & 215,487 & 121         & 1        & 129.1$\pm$1.4                   & 0.1$\pm$0.03  & 1.79$\pm$0.19 \\
 &                      & 0.1  & 49,471  & 177  & 1 & 155.2$\pm$2.8 & 0.15$\pm$0.05 & 1.08$\pm$0.14 \\
 &                      & 0.05 & 23,704  & 266  & 2 & 139.8$\pm$1.4 & 0.15$\pm$0.06 & 1.2$\pm$0.16  \\
 &                      & 0.01 & 13,158  & 592  & 4 & 135.6$\pm$1.1 & 0.31$\pm$0.08 & 2.02$\pm$0.18 \\ \clineB{2-9}{2}
 & \multirow{5}{*}{\kpar} & -    & -      & 101  & 1 & 318$\pm$79.5  & 0.05$\pm$0    & 0.27$\pm$0.05 \\
 &                      & -    & -      & 201  & 2 & 169$\pm$6.2   & 0.2$\pm$0     & 0.41$\pm$0.02 \\
 &                      & -    & -      & 301  & 3 & 153$\pm$3.3   & 0.42$\pm$0.01 & 0.67$\pm$0.05 \\
 &                      & -    & -      & 401  & 4 & 144$\pm$2     & 0.66$\pm$0.01 & 0.94$\pm$0.07 \\
 &                      & -    & -      & 501  & 5 & 140$\pm$1.4   & 0.97$\pm$0.02 & 1.3$\pm$0.07  \\ \hlineB{2}
\multirow{9}{*}{100} & \multirow{4}{*}{\mainalgname} & 0.2     & 478,318 & 177         & 1        & 99.6$\pm$0.8                    & 0.22$\pm$0.08 & 4.27$\pm$0.34 \\
 &                      & 0.1  & 109,813 & 283  & 1 & 126.8$\pm$3.1 & 0.18$\pm$0.06 & 1.9$\pm$0.18  \\
 &                      & 0.05 & 52,616  & 378  & 2 & 110.5$\pm$1.3 & 0.17$\pm$0.06 & 2.22$\pm$0.16 \\
                     &                         & 0.01    & 29,207  & 722$\pm$32     & 3.1$\pm$0.3 & 105.3$\pm$1.4                   & 0.39$\pm$0.09 & 2.98$\pm$0.18 \\ \clineB{2-9}{2}
 & \multirow{5}{*}{\kpar} & -    & -      & 201  & 1 & 264$\pm$67.1  & 0.05$\pm$0    & 0.44$\pm$0.02 \\  
 &                      & -    & -      & 401  & 2 & 133$\pm$2.3   & 0.31$\pm$0.01 & 0.77$\pm$0.08 \\
 &                      & -    & -      & 601  & 3 & 119$\pm$1     & 0.65$\pm$0.01 & 1.15$\pm$0.06 \\
 &                      & -    & -      & 801  & 4 & 112$\pm$1     & 1.1$\pm$0.01  & 1.69$\pm$0.13 \\
 &                      & -    & -      & 1001 & 5 & 109$\pm$0.7   & 1.67$\pm$0.03 & 2.21$\pm$0.08 \\ \hlineB{2}
\multirow{8}{*}{200} & \multirow{3}{*}{\mainalgname} & 0.1     & 241,364 & 489         & 1        & 107.8$\pm$2.7                   & 0.31$\pm$0.06 & 4.33$\pm$0.43 \\
 &                      & 0.05 & 115,648 & 592  & 2 & 89.4$\pm$0.6  & 0.29$\pm$0.07 & 4.48$\pm$0.26 \\
 &                      & 0.01 & 64,197  & 1130 & 3 & 87.3$\pm$0.2  & 0.47$\pm$0.08 & 5.49$\pm$0.34 \\ \clineB{2-9}{2}
 & \multirow{5}{*}{\kpar} & -    & -      & 401  & 1 & 224$\pm$49.2  & 0.05$\pm$0    & 0.92$\pm$0.14 \\
 &                      & -    & -      & 801  & 2 & 104$\pm$2.7   & 0.5$\pm$0.01  & 1.45$\pm$0.07 \\
 &                      & -    & -      & 1201 & 3 & 93.8$\pm$1    & 1.14$\pm$0.02 & 2.27$\pm$0.09 \\
 &                      & -    & -      & 1601 & 4 & 88.7$\pm$0.4  & 1.96$\pm$0.03 & 3.19$\pm$0.06 \\
 &                      & -    & -      & 2001 & 5 & 87$\pm$0.6    & 3.03$\pm$0.06 & 4.35$\pm$0.1  \\ \hlineB{2}
\end{tabular}
\end{table}

\begin{table}[]
\caption{\texttt{KDDCup1999} dataset experiments with MiniBatchKMeans used as black-box.
‘T’ stands for time in seconds.}
\label{tab:experiments_summary_KDD_S}
\begin{tabular}{lllllllll} \hline 
$k$                  & Alg                     & $\epsilon$ & $P_{1}$   & Output size & Rounds    & Cost ($\cdot 10^{10}$) & T (machine)  & T (Total)    \\ \hlineB{3}
\multirow{9}{*}{25}  & \multirow{4}{*}{\mainalgname} & 0.2     & 110,088 & 115         & 1         & 6,273,229$\pm$102,539& 0.14$\pm$0.01 & 0.94$\pm$0.15 \\
                     &                         & 0.1     & 23,590  & 313         & 3         & 3,034,310$\pm$1,064,972          & 0.25$\pm$0.01 & 1.33$\pm$0.13 \\
                     &                         & 0.05    & 10,920  & 433         & 4         & 3,982,745$\pm$1124135            & 0.37$\pm$0.04 & 1.6$\pm$0.2   \\
                     &                         & 0.01    & 5,896   & 1243$\pm$61    & 10.5$\pm$0.5 & 2,344,966$\pm$946,061      & 0.88$\pm$0.06 & 3.88$\pm$0.21 \\ \clineB{2-9}{2}
                     & \multirow{5}{*}{\kpar}    & -       & -      & 51          & 1         & 254$\pm$35                       & 0.07$\pm$0    & 0.18$\pm$0.03 \\
                     &                         & -       & -      & 101         & 2         & 157$\pm$12.3                     & 0.23$\pm$0    & 0.34$\pm$0.03 \\
                     &                         & -       & -      & 151         & 3         & 148$\pm$22.3                     & 0.44$\pm$0.01 & 0.55$\pm$0.03 \\
                     &                         & -       & -      & 201         & 4         & 124$\pm$3.3                      & 0.71$\pm$0.01 & 0.82$\pm$0.04 \\
                     &                         & -       & -      & 251         & 5         & 126$\pm$11.8                     & 1.03$\pm$0.01 & 1.15$\pm$0.02 \\ \hlineB{2}
\multirow{9}{*}{50}  & \multirow{4}{*}{\mainalgname} & 0.2     & 247,347 & 171         & 1         & 5,971,765$\pm$351183           & 0.17$\pm$0.01 & 1.59$\pm$0.18 \\
                     &                         & 0.1     & 53,003  & 304         & 2         & 5,716,218$\pm$507,248           & 0.29$\pm$0.01 & 1.5$\pm$0.13  \\
                     &                         & 0.05    & 24,535  & 555$\pm$56     & 3.8$\pm$0.4  & 2,995,765$\pm$1,379,745        & 0.43$\pm$0.03 & 2.05$\pm$0.21 \\
                     &                         & 0.01    & 13,249  & 1248$\pm$47    & 8.1$\pm$0.3  & 3,946,555$\pm$954,957           & 0.89$\pm$0.02 & 3.86$\pm$0.23 \\ \clineB{2-9}{2}
                     & \multirow{5}{*}{\kpar}    & -       & -      & 101         & 1         & 109$\pm$29.1                     & 0.07$\pm$0    & 0.25$\pm$0.02 \\
                     &                         & -       & -      & 201         & 2         & 37.2$\pm$11                      & 0.33$\pm$0.01 & 0.56$\pm$0.08 \\
                     &                         & -       & -      & 301         & 3         & 35.2$\pm$4.4                     & 0.69$\pm$0.01 & 0.91$\pm$0.05 \\
                     &                         & -       & -      & 401         & 4         & 35.2$\pm$5.3                     & 1.13$\pm$0.01 & 1.36$\pm$0.04 \\
                     &                         & -       & -      & 501         & 5         & 32.8$\pm$5.6                     & 1.69$\pm$0.02 & 1.97$\pm$0.08 \\ \hlineB{2}
\multirow{9}{*}{100} & \multirow{4}{*}{\mainalgname} & 0.2     & 549,037 & 277         & 1         & 5,969,280$\pm$355,576& 0.26$\pm$0.02 & 3.22$\pm$0.22 \\
                     &                         & 0.1     & 117,651 & 466         & 2         & 5,587,433$\pm$940,351 & 0.38$\pm$0.02 & 2.66$\pm$0.27 \\
                     &                         & 0.05    & 54,461  & 667         & 3         & 3,924,526$\pm$674,470           & 0.52$\pm$0.02 & 2.89$\pm$0.18 \\
                     &                         & 0.01    & 29,409  & 1528        & 7         & 3,530,365$\pm$841,054           & 1.11$\pm$0.03 & 5.36$\pm$0.25 \\ \clineB{2-9}{2}
                     & \multirow{5}{*}{\kpar}    & -       & -      & 201         & 1         & 51.7$\pm$15.6                    & 0.06$\pm$0    & 0.47$\pm$0.05 \\
                     &                         & -       & -      & 401         & 2         & 6.5$\pm$0.7                      & 0.54$\pm$0.01 & 0.92$\pm$0.04 \\
                     &                         & -       & -      & 601         & 3         & 8.4$\pm$0.8                      & 1.16$\pm$0.02 & 1.62$\pm$0.14 \\
                     &                         & -       & -      & 801         & 4         & 8.5$\pm$1.3                      & 1.99$\pm$0.03 & 2.4$\pm$0.04  \\
                     &                         & -       & -      & 1001        & 5         & 8$\pm$0.5                        & 3.05$\pm$0.03 & 3.55$\pm$0.13 \\ \hlineB{2}
\multirow{8}{*}{200} & \multirow{3}{*}{\mainalgname} & 0.1     & 258,592 & 778         & 2         & 4,561,690$\pm$1,036,461       & 0.61$\pm$0.04 & 5.3$\pm$0.25  \\
                     &                         & 0.05    & 119,705 & 1088        & 3         & 2,859,826$\pm$1,177,854            & 0.78$\pm$0.07 & 5.2$\pm$0.21  \\
                     &                         & 0.01    & 64,641  & 2060        & 6         & 3,814,931$\pm$1,066,684          & 1.46$\pm$0.08 & 8.22$\pm$0.23 \\ \clineB{2-9}{2}
                     & \multirow{5}{*}{\kpar}    & -       & -      & 401         & 1         & 10.8$\pm$2.1                     & 0.06$\pm$0    & 0.89$\pm$0.18 \\
                     &                         & -       & -      & 801         & 2         & 1.7$\pm$0.5                      & 0.92$\pm$0.02 & 1.81$\pm$0.23 \\
                     &                         & -       & -      & 1201        & 3         & 2.6$\pm$0.4                      & 2.11$\pm$0.05 & 3.06$\pm$0.25 \\
                     &                         & -       & -      & 1601        & 4         & 2.8$\pm$0.1                      & 3.64$\pm$0.09 & 4.6$\pm$0.12  \\
                     &                         & -       & -      & 2001        & 5         & 2.4$\pm$0.4                      & 5.69$\pm$0.08 & 6.73$\pm$0.11 \\ \hlineB{2}
\end{tabular}
\end{table}

\begin{table}[]
\caption{\texttt{BigCross} dataset experiments with MiniBatchKMeans used as black-box.
‘T’ stands for time in seconds.}
\label{tab:experiments_summary_bigcross_S}
\begin{tabular}{lllllllll} \hline 
$k$                  & ALG                     & epsilon & $P_1$      & Output size & Rounds   & Cost (10\textasciicircum{}10) & T (machine)   & T (Total)     \\ \hlineB{3}
\multirow{9}{*}{25}  & \multirow{4}{*}{\mainalgname} & 0.2     & 126,978 & 90          & 1        & 328$\pm$2.6                      & 0.38$\pm$0.03  & 1.39$\pm$0.13  \\
                     &                         & 0.1     & 25,335   & 99$\pm$8       & 1        & 327$\pm$6.3                      & 0.37$\pm$0.05  & 0.86$\pm$0.12  \\
                     &                         & 0.05    & 11,316   & 204         & 2        & 345$\pm$8                        & 0.4$\pm$0.03   & 1.1$\pm$0.15   \\
                     &                         & 0.01    & 5,939    & 365$\pm$12     & 3        & 318$\pm$3.2                      & 0.79$\pm$0.02  & 1.81$\pm$0.11  \\ \clineB{2-9}{2}
                     & \multirow{5}{*}{\kpar}    & -       & -       & 51          & 1        & 519$\pm$39.6                     & 0.18$\pm$0.01  & 0.27$\pm$0.03  \\
                     &                         & -       & -       & 101         & 2        & 367$\pm$8.7                      & 0.49$\pm$0.01  & 0.6$\pm$0.02   \\
                     &                         & -       & -       & 151         & 3        & 350$\pm$4.8                      & 0.93$\pm$0.04  & 1.05$\pm$0.06  \\
                     &                         & -       & -       & 201         & 4        & 339$\pm$6.3                      & 1.59$\pm$0.05  & 1.71$\pm$0.06  \\
                     &                         & -       & -       & 251         & 5        & 330$\pm$5.5                      & 2.14$\pm$0.09  & 2.28$\pm$0.1   \\ \hlineB{2}
\multirow{9}{*}{50}  & \multirow{4}{*}{\mainalgname} & 0.2     & 285,296  & 121         & 1        & 222$\pm$3.7                      & 0.47$\pm$0.05  & 2.44$\pm$0.23  \\
                     &                         & 0.1     & 56,924   & 142$\pm$24     & 1        & 221$\pm$2.1                      & 0.45$\pm$0.02  & 1.3$\pm$0.17   \\
                     &                         & 0.05    & 25,427   & 266         & 2        & 242$\pm$4.2                      & 0.48$\pm$0.04  & 1.55$\pm$0.19  \\
                     &                         & 0.01    & 13,344   & 444         & 3        & 217$\pm$1.8                      & 0.85$\pm$0.13  & 2.26$\pm$0.2   \\ \clineB{2-9}{2}
                     & \multirow{5}{*}{\kpar}    & -       & -       & 101         & 1        & 365$\pm$28                       & 0.18$\pm$0.02  & 0.39$\pm$0.03  \\
                     &                         & -       & -       & 201         & 2        & 244$\pm$5.9                      & 0.78$\pm$0.07  & 1.02$\pm$0.09  \\
                     &                         & -       & -       & 301         & 3        & 230$\pm$2.1                      & 1.33$\pm$0.01  & 1.58$\pm$0.03  \\
                     &                         & -       & -       & 401         & 4        & 223$\pm$2.3                      & 2.27$\pm$0.07  & 2.55$\pm$0.1   \\
                     &                         & -       & -       & 501         & 5        & 217$\pm$1.4                      & 3.54$\pm$0.19  & 3.84$\pm$0.12  \\ \hlineB{2}
\multirow{9}{*}{100} & \multirow{4}{*}{\mainalgname} & 0.2     & 633,272  & 177         & 1        & 150$\pm$1.3                      & 0.59$\pm$0.05  & 5.17$\pm$0.23  \\
                     &                         & 0.1     & 126,354  & 193         & 1        & 148$\pm$1.2                      & 0.59$\pm$0.07  & 2.3$\pm$0.15   \\
                     &                         & 0.05    & 56,440   & 298$\pm$28     & 1.1$\pm$0.3 & 169$\pm$2.7                      & 0.59$\pm$0.05  & 2.15$\pm$0.18  \\
                     &                         & 0.01    & 29,620   & 612         & 3        & 154$\pm$1.3                      & 0.87$\pm$0.03  & 3.1$\pm$0.17   \\ \clineB{2-9}{2}
                     & \multirow{5}{*}{\kpar}    & -       & -       & 201         & 1        & 242$\pm$20                       & 0.18$\pm$0.03  & 0.56$\pm$0.06  \\
                     &                         & -       & -       & 401         & 2        & 169$\pm$2.3                      & 1.1$\pm$0.04   & 1.54$\pm$0.05  \\
                     &                         & -       & -       & 601         & 3        & 157$\pm$2                        & 2.43$\pm$0.18  & 2.94$\pm$0.18  \\
                     &                         & -       & -       & 801         & 4        & 153$\pm$1.5                      & 3.94$\pm$0.25  & 4.49$\pm$0.33  \\
                     &                         & -       & -       & 1001        & 5        & 150$\pm$1                        & 6.17$\pm$0.31  & 6.71$\pm$0.32  \\ \hlineB{2}
\multirow{8}{*}{200} & \multirow{3}{*}{\mainalgname} & 0.1     & 277,721  & 289         & 1        & 102$\pm$0.3                      & 0.74$\pm$0.02  & 4.33$\pm$0.25  \\
                     &                         & 0.05    & 124,053  & 496         & 1        & 116$\pm$2.1                      & 0.79$\pm$0.07  & 3.72$\pm$0.16  \\
                     &                         & 0.01    & 65,104   & 820         & 2        & 109$\pm$0.7                      & 1.13$\pm$0.05  & 4.74$\pm$0.18  \\ \clineB{2-9}{2}
                     & \multirow{5}{*}{\kpar}    & -       & -       & 401         & 1        & 166$\pm$8.6                      & 0.21$\pm$0.06  & 1.04$\pm$0.07  \\
                     &                         & -       & -       & 801         & 2        & 119$\pm$1.4                      & 1.8$\pm$0.03   & 2.72$\pm$0.03  \\
                     &                         & -       & -       & 1201        & 3        & 111$\pm$1                        & 4.08$\pm$0.13  & 5.31$\pm$0.28  \\
                     &                         & -       & -       & 1601        & 4        & 107$\pm$0.5                      & 7.31$\pm$0.09  & 8.56$\pm$0.09  \\
                     &                         & -       & -       & 2001        & 5        & 106$\pm$0.4                      & 10.86$\pm$0.32 & 12.24$\pm$0.32 \\ \hlineB{2}
\end{tabular}
\end{table}

\end{document}

%% file: prel.tex
\usepackage{amsmath,amsfonts,stmaryrd,amsthm,amssymb,enumitem}
\usepackage{algorithm,booktabs,url, hyperref}

\usepackage{algorithmic}

\usepackage{pifont}

\usepackage{tikz}
\usetikzlibrary{arrows}
\usetikzlibrary{calc}
\usetikzlibrary{shapes}
\usetikzlibrary{fit}
\usetikzlibrary{matrix} 
\usetikzlibrary{positioning}
\usetikzlibrary{decorations.pathreplacing}

\newtheorem{theorem}{Theorem}[section]
\newtheorem{lemma}[theorem]{Lemma}
\newtheorem{cor}[theorem]{Corollary}

\renewcommand{\eqref}[1]{Eq.~(\ref{eq:#1})}

\newcommand{\tabref}[1]{Table~\ref{tab:#1}}        
\newcommand{\secref}[1]{Section~\ref{sec:#1}}
\newcommand{\thmref}[1]{Theorem~\ref{thm:#1}}
\newcommand{\lemref}[1]{Lemma~\ref{lem:#1}}

\newcommand{\corref}[1]{Cor.~\ref{cor:#1}}
\newcommand{\appref}[1]{Appendix~\ref{app:#1}}

\newcommand{\algref}[1]{Alg.~\ref{alg:#1}}

\input{prel_def.tex}

\input{prel_cal.tex}

%% file: prel_def.tex
\renewcommand{\P}{\mathbb{P}}

\newcommand{\reals}{\mathbb{R}}
\newcommand{\nats}{\mathbb{N}}

\newcommand{\ceil}[1]{{\lceil #1\rceil}}

\DeclareMathOperator*{\argmin}{argmin}



%% file: prel_cal.tex
\newcommand{\cA}{\mathcal{A}}

%% file: Arxiv_paper.bbl
\begin{thebibliography}{25}
\providecommand{\natexlab}[1]{#1}
\providecommand{\url}[1]{\texttt{#1}}
\expandafter\ifx\csname urlstyle\endcsname\relax
  \providecommand{\doi}[1]{doi: #1}\else
  \providecommand{\doi}{doi: \begingroup \urlstyle{rm}\Url}\fi

\bibitem[Ackermann et~al.(2012)Ackermann, M{\"a}rtens, Raupach, Swierkot,
  Lammersen, and Sohler]{ackermann2012streamkm++}
Marcel~R Ackermann, Marcus M{\"a}rtens, Christoph Raupach, Kamil Swierkot,
  Christiane Lammersen, and Christian Sohler.
\newblock Stream{KM}++: A clustering algorithm for data streams.
\newblock \emph{Journal of Experimental Algorithmics (JEA)}, 17:\penalty0 2--4,
  2012.

\bibitem[Ahmadian et~al.(2019)Ahmadian, Norouzi-Fard, Svensson, and
  Ward]{ahmadian2019better}
Sara Ahmadian, Ashkan Norouzi-Fard, Ola Svensson, and Justin Ward.
\newblock Better guarantees for k-means and euclidean k-median by primal-dual
  algorithms.
\newblock \emph{SIAM Journal on Computing}, 49\penalty0 (4):\penalty0
  FOCS17--97, 2019.

\bibitem[Ailon et~al.(2009)Ailon, Jaiswal, and Monteleoni]{ailon2009streaming}
Nir Ailon, Ragesh Jaiswal, and Claire Monteleoni.
\newblock Streaming k-means approximation.
\newblock In \emph{Advances in neural information processing systems}, pages
  10--18, 2009.

\bibitem[Bachem et~al.(2017{\natexlab{a}})Bachem, Lucic, and
  Krause]{bachem2017distributed}
Olivier Bachem, Mario Lucic, and Andreas Krause.
\newblock Distributed and provably good seedings for k-means in constant
  rounds.
\newblock In \emph{International Conference on Machine Learning}, pages
  292--300. PMLR, 2017{\natexlab{a}}.

\bibitem[Bachem et~al.(2017{\natexlab{b}})Bachem, Lucic, and
  Krause]{bachem2017practical}
Olivier Bachem, Mario Lucic, and Andreas Krause.
\newblock Practical coreset constructions for machine learning.
\newblock \emph{arXiv preprint arXiv:1703.06476}, 2017{\natexlab{b}}.

\bibitem[Bahmani et~al.(2012)Bahmani, Moseley, Vattani, Kumar, and
  Vassilvitskii]{bahmani2012scalable}
Bahman Bahmani, Benjamin Moseley, Andrea Vattani, Ravi Kumar, and Sergei
  Vassilvitskii.
\newblock Scalable k-means++.
\newblock \emph{arXiv preprint arXiv:1203.6402}, 2012.

\bibitem[Balcan et~al.(2013)Balcan, Ehrlich, and Liang]{balcan2013distributed}
Maria~Florina Balcan, Steven Ehrlich, and Yingyu Liang.
\newblock Distributed k-means and k-median clustering on general topologies.
\newblock \emph{arXiv preprint arXiv:1306.0604}, 2013.

\bibitem[Baldi et~al.(2014)Baldi, Sadowski, and Whiteson]{baldi2014searching}
Pierre Baldi, Peter Sadowski, and Daniel Whiteson.
\newblock Searching for exotic particles in high-energy physics with deep
  learning.
\newblock \emph{Nature communications}, 5\penalty0 (1):\penalty0 1--9, 2014.

\bibitem[Bhaskara and Wijewardena(2018)]{bhaskara2018distributed}
Aditya Bhaskara and Maheshakya Wijewardena.
\newblock Distributed clustering via lsh based data partitioning.
\newblock In \emph{International Conference on Machine Learning}, pages
  570--579. PMLR, 2018.

\bibitem[Chen et~al.(2016)Chen, Sun, Woodruff, and
  Zhang]{chen2016communication}
Jiecao Chen, He~Sun, David Woodruff, and Qin Zhang.
\newblock Communication-optimal distributed clustering.
\newblock \emph{Advances in Neural Information Processing Systems},
  29:\penalty0 3727--3735, 2016.

\bibitem[Chen et~al.(2018)Chen, Azer, and Zhang]{chen2018practical}
Jiecao Chen, Erfan~Sadeqi Azer, and Qin Zhang.
\newblock A practical algorithm for distributed clustering and outlier
  detection.
\newblock \emph{arXiv preprint arXiv:1805.09495}, 2018.

\bibitem[Dua and Graff(2017)]{Dua:2019}
Dheeru Dua and Casey Graff.
\newblock {UCI} machine learning repository, 2017.
\newblock URL \url{http://archive.ics.uci.edu/ml}.

\bibitem[Ene et~al.(2011)Ene, Im, and Moseley]{ene2011fast}
Alina Ene, Sungjin Im, and Benjamin Moseley.
\newblock Fast clustering using mapreduce.
\newblock In \emph{Proceedings of the 17th ACM SIGKDD international conference
  on Knowledge discovery and data mining}, pages 681--689, 2011.

\bibitem[Feldman et~al.(2020)Feldman, Schmidt, and Sohler]{feldman2020turning}
Dan Feldman, Melanie Schmidt, and Christian Sohler.
\newblock Turning big data into tiny data: Constant-size coresets for k-means,
  pca, and projective clustering.
\newblock \emph{SIAM Journal on Computing}, 49\penalty0 (3):\penalty0 601--657,
  2020.

\bibitem[Guha et~al.(2003)Guha, Meyerson, Mishra, Motwani, and
  O'Callaghan]{guha2003clustering}
Sudipto Guha, Adam Meyerson, Nina Mishra, Rajeev Motwani, and Liadan
  O'Callaghan.
\newblock Clustering data streams: Theory and practice.
\newblock \emph{IEEE transactions on knowledge and data engineering},
  15\penalty0 (3):\penalty0 515--528, 2003.

\bibitem[Guha et~al.(2019)Guha, Li, and Zhang]{guha2019distributed}
Sudipto Guha, Yi~Li, and Qin Zhang.
\newblock Distributed partial clustering.
\newblock \emph{ACM Transactions on Parallel Computing (TOPC)}, 6\penalty0
  (3):\penalty0 1--20, 2019.

\bibitem[Guo and Li(2018)]{guo2018distributed}
Xiangyu Guo and Shi Li.
\newblock Distributed $ k $-clustering for data with heavy noise.
\newblock \emph{arXiv preprint arXiv:1810.07852}, 2018.

\bibitem[Hess et~al.(2021)Hess, Moshkovitz, and Sabato]{hess2021constant}
Tom Hess, Michal Moshkovitz, and Sivan Sabato.
\newblock A constant approximation algorithm for sequential no-substitution
  k-median clustering under a random arrival order.
\newblock \emph{arXiv preprint arXiv:2102.04050}, 2021.

\bibitem[Kumar et~al.(2015)Kumar, Moseley, Vassilvitskii, and
  Vattani]{kumar2015fast}
Ravi Kumar, Benjamin Moseley, Sergei Vassilvitskii, and Andrea Vattani.
\newblock Fast greedy algorithms in mapreduce and streaming.
\newblock \emph{ACM Transactions on Parallel Computing (TOPC)}, 2\penalty0
  (3):\penalty0 1--22, 2015.

\bibitem[Laurent and Massart(2000)]{laurent2000adaptive}
Beatrice Laurent and Pascal Massart.
\newblock Adaptive estimation of a quadratic functional by model selection.
\newblock \emph{Annals of Statistics}, pages 1302--1338, 2000.

\bibitem[Meng et~al.(2016)Meng, Bradley, Yavuz, Sparks, Venkataraman, Liu,
  Freeman, Tsai, Amde, Owen, et~al.]{meng2016mllib}
Xiangrui Meng, Joseph Bradley, Burak Yavuz, Evan Sparks, Shivaram Venkataraman,
  Davies Liu, Jeremy Freeman, DB~Tsai, Manish Amde, Sean Owen, et~al.
\newblock Mllib: Machine learning in apache spark.
\newblock \emph{The Journal of Machine Learning Research}, 17\penalty0
  (1):\penalty0 1235--1241, 2016.

\bibitem[Motwani and Raghavan(1996)]{motwani1996randomized}
Rajeev Motwani and Prabhakar Raghavan.
\newblock Randomized algorithms.
\newblock \emph{ACM Computing Surveys (CSUR)}, 28\penalty0 (1):\penalty0
  33--37, 1996.

\bibitem[Pedregosa et~al.(2011)Pedregosa, Varoquaux, Gramfort, Michel, Thirion,
  Grisel, Blondel, Prettenhofer, Weiss, Dubourg, Vanderplas, Passos,
  Cournapeau, Brucher, Perrot, and Duchesnay]{scikit-learn}
F.~Pedregosa, G.~Varoquaux, A.~Gramfort, V.~Michel, B.~Thirion, O.~Grisel,
  M.~Blondel, P.~Prettenhofer, R.~Weiss, V.~Dubourg, J.~Vanderplas, A.~Passos,
  D.~Cournapeau, M.~Brucher, M.~Perrot, and E.~Duchesnay.
\newblock Scikit-learn: Machine learning in {P}ython.
\newblock \emph{Journal of Machine Learning Research}, 12:\penalty0 2825--2830,
  2011.

\bibitem[Sebestyen(1962)]{sebestyen1962decision}
George~S Sebestyen.
\newblock \emph{Decision-making processes in pattern recognition (ACM monograph
  series)}.
\newblock Macmillan Publishing Co., Inc., 1962.

\bibitem[Voevodski(2021)]{voevodski2021large}
Konstantin Voevodski.
\newblock Large scale k-median clustering for stable clustering instances.
\newblock In \emph{International Conference on Artificial Intelligence and
  Statistics}, pages 2890--2898. PMLR, 2021.

\end{thebibliography}
